\documentclass[DIV=11,a4paper]{artclcomp}

\usepackage{amsmath}
\usepackage{amssymb}
\usepackage{times}
\usepackage{calc}
\usepackage{enumitem}
\usepackage{eurosym}
\usepackage{textcomp}
\usepackage{mathrsfs}
\usepackage{graphicx}
\usepackage{enumitem}
\usepackage{stmaryrd}
\usepackage{epsfig}
\usepackage[usenames,dvipsnames]{xcolor}

\usepackage{theoremref} 

\newcommand{\R}{\mathbb{R}}

\usepackage[hyperref,amsmath,thmmarks]{ntheorem}

\theoremseparator{.}

\theoremsymbol{\ensuremath{\blacklozenge}}
\theoremheaderfont{\itshape}

\numberwithin{equation}{section}

\usepackage{lipsum}
\usepackage{graphicx}
\usepackage{amsfonts}
\usepackage{accents}
\usepackage{xspace}
\usepackage{dsfont}
\usepackage{bbm} 
\usepackage{subcaption}         
\usepackage{theoremref}
\usepackage[mathscr]{euscript}

\allowdisplaybreaks[1]

\def\FH{Fr\'echet--Hoeffding\xspace}

\def\ud{\mathrm{d}}

\def\supp{\ensuremath{\mathrm{supp}}}

\DeclareMathAccent{\what}{\mathord}{largesymbols}{"62}

\DeclareFontFamily{U}{mathx}{\hyphenchar\font45}
\DeclareFontShape{U}{mathx}{m}{n}{
      <5> <6> <7> <8> <9> <10>
      <10.95> <12> <14.4> <17.28> <20.74> <24.88>
      mathx10
      }{}
\DeclareSymbolFont{mathx}{U}{mathx}{m}{n}
\DeclareFontSubstitution{U}{mathx}{m}{n}
\DeclareMathAccent{\widecheck}{0}{mathx}{"71}

\newcommand*{\bigtimes}{\mathop{\raisebox{-.5ex}{\hbox{\huge{$\times$}}}}}

\begin{document}

\title{Detection of arbitrage opportunities in multi-asset\newline derivatives markets}

\author[a,b,1,s]{Antonis Papapantoleon}
\author[c,2,s]{Paulo Yanez Sarmiento}

\address[a]{{Delft Institute of Applied Mathematics, TU Delft, 2628 Delft, The Netherlands}}
\address[b]{Institute of Applied and Computational Mathematics, FORTH, 70013 Heraklion, Greece}
\address[c]{Institute of Mathematics, TU Berlin, Stra\ss e des 17. Juni 136, 10623 Berlin, Germany}

\eMail[1]{a.papapantoleon@tudelft.nl}
\eMail[2]{p.yanez@outlook.de}

\myThanks[s]{We thank Thibaut Lux for fruitful discussions during the work on these topics. AP gratefully acknowledges the financial support from the Hellenic Foundation for Research and Innovation Grant No. HFRI-FM17-2152.}

\abstract{
We are interested in the existence of equivalent martingale measures and the detection of arbitrage opportunities in markets where several multi-asset derivatives are traded simultaneously.
More specifically, we consider a financial market with multiple traded assets whose marginal risk-neutral distributions are known, and assume that several derivatives written on these assets are traded simultaneously.
In this setting, there is a bijection between the existence of an equivalent martingale measure and the existence of a copula that couples these marginals.
Using this bijection and recent results on improved \FH bounds in the presence of additional information on functionals of a copula by \citet{lux2016}, we can extend the results of \citet{Tav15} on the detection of arbitrage opportunities to the general multi-dimensional case.
More specifically, we derive sufficient conditions for the absence of arbitrage and formulate an optimization problem for the detection of a possible arbitrage opportunity.
This problem can be solved efficiently using numerical optimization routines.
The most interesting practical outcome is the following: we can construct a financial market where each multi-asset derivative is traded within its own no-arbitrage interval, and yet when considered together an arbitrage opportunity may arise.
}

\keyWords{Arbitrage, equivalent martingale measures, detection of arbitrage opportunities, multiple assets, multi-asset derivatives, copulas, improved \FH bounds.}

\keyAMSClassification{91G20, 62H05, 60E15.}
\keyJELClassification{}


\date{} \maketitle \frenchspacing 


\section{Introduction}

We consider a financial market where multiple assets and several derivatives written on single or multiple assets are traded simultaneously. 
Assuming we are given a set of traded prices for these multi-asset derivatives, we are interested in whether there exists an arbitrage-free model that is consistent with these prices or not.
A consistent arbitrage-free model will exist if we can find an equivalent martingale measure such that we can describe these prices as discounted expected payoffs under this measure.
We assume that the marginal risk-neutral distributions of the assets are known, \textit{e.g.} they have been estimated from single-asset options prices using \citet{breeden}.
Then, there exists a bijection between the existence of an equivalent martingale measure and the existence of a copula that couples these marginal distributions.
Using recent results about improved \FH bounds on copulas in the presence of additional information, we can formulate a sufficient condition for the existence of a copula and thus for the absence of arbitrage in this financial market.
Moreover, the formulation of this condition as an optimization problem allows for the detection of an arbitrage opportunity via numerical optimization routines. 

Arbitrage is a fundamental concept in economics and finance, because the modern theory of option valuation is rooted on the assumption of the absence of arbitrage, while it is also closely related with notions of equilibrium in financial markets.
Arbitrage is also a concept of practical importance, as financial institutions are interested in ensuring that their systems for option valuation, simulation, scenario generation, \textit{etc}, are free of arbitrage, in order to be useful and relevant.
Therefore, topics related to the existence of arbitrage and the consistency of arbitrage-free models with given traded prices are of significant theoretical and practical interest.

There is a sufficiently rich literature by now devoted to the case where a single asset and options on this asset are traded in a financial market.
\citet{Laurent_Leisen_2000} in their pioneering work provide a procedure to check for the absence of arbitrage in a discrete set of market data.
\citet{Carr_Madan_2005} provide a sufficient condition for the absence of arbitrage in a market where countably-infinite many European options with discrete strikes can be traded.
These results where later generalized and extended by \citet{Cousot_2007}, by \citet{Buehler_2006}, and in particular by \citet{Davis_Hobson_2007} who provided necessary and sufficient conditions for the existence of an arbitrage-free model consistent with a set of market prices.
More recently, \citet{Gerhold_Gulum_2020} considered the same problem in case the only observables are the bid and ask prices of the underlying asset.

The literature is not that developed when one turns to multiple underlying assets and multi-asset derivatives. 
Actually, to the best of our knowledge, the only work treating this problem until very recently was \citet{Tav15}.
The setting here is exactly the same as in \citet{Tav15}, \textit{i.e.} the author considers multiple underlying assets with known risk-neutral marginals and several traded derivatives on multiple assets, and provides two methods for detecting arbitrage opportunities, one based on Bernstein copulas and another based on improved \FH bounds, which is however restricted to the two-asset case. 
In our work, we extend the results of \citet{Tav15} to the general multi-asset case using the recent results of \citet{lux2016} on improved \FH bounds for $d$-copulas in the presence of additional information on functionals of a copula, with $d\ge2$.

In a very recent work, \citet{Neufeld_Papapantoleon_Xiang_2020} developed numerical methods for the computation of model-free bounds for multi-asset option prices in the presence of other traded multi-asset derivatives.
Their results are supported by a Fundamental Theorem of Asset Pricing for their setting, and the numerical methods are used for the detection of arbitrage opportunities as well.
Compared to that work, our method is conceptually and computationally simpler, while both methods are numerically efficient.
However, our method only yields the existence of an arbitrage opportunity, while the method of \citet{Neufeld_Papapantoleon_Xiang_2020} also delivers the portfolio that generates the arbitrage.
Moreover, the two settings are not the same; here we assume that the marginals are known, while in \citet{Neufeld_Papapantoleon_Xiang_2020} the authors assume the existence of finitely many traded European options for each single asset.
Therefore, the two methods are not directly comparable.

Let us also mention that there exist several articles on the computation of model-free bounds for multi-asset option prices; see \textit{e.g.} \citet{Dhaene_etal_2002_a,Dhaene_etal_2002_b}, \citet{Hobson_Laurence_Wang_2005_1,Hobson_Laurence_Wang_2005_2}, \citet{dAspremont_ElGhaoui_2006}, \citet{Pena_Vera_Zuluaga_2010}, \citet{tankov} and \citet{lux2016} to mention a fraction of this literature, while we refer the reader to \cite{lux2016} and \cite{Neufeld_Papapantoleon_Xiang_2020} for more comprehensive literature reviews.
These methods could also be used, in principle, to detect arbitrage opportunities, although this has not been pursued yet. 

The remainder of this article is structured as follows:
In Section \ref{sec:main-properties} we review some necessary results about copulas, quasi-copulas and improved \FH bounds.
In Section \ref{sec:ISD-CQ} we present results on integration and stochastic dominance for quasi-copulas; these include also a new representation of the integral with respect to a quasi-copula that could be of independent interest.
In Section \ref{sec:CA} we revisit the bijection between the existence of an equivalent martingale measure and a copula that couples the marginals of the underlying assets already present in \citet{Tav15}, and derive necessary conditions for the absence of arbitrage in the presence of several multi-asset derivatives traded simultaneously.
In Section \ref{sec:appl} we apply our results in models with three underlying assets. 
In particular, we show that we can construct a financial market where each multi-asset derivative is traded within its own no-arbitrage interval, and yet when considered together an arbitrage opportunity may arise.
Finally, the appendix collects some additional results.


\section{Copulas, quasi-copulas and improved \FH bounds}
\label{sec:main-properties}

This section serves as an introduction to the notation that will be used throughout this work, as well as to the basic notion of copulas, quasi-copulas and (improved) \FH bounds.
Let $d\ge2$ be an integer, and set $\mathbb{I}=[0,1]$ and $\mathbf 1=(1,\dots,1)\in\R^d$.
In the sequel, boldface letters, such as $\mathbf{u}$ or $\mathbf{v}$, denote vectors in $\mathbb I^d$ or $\mathbb{R}^d$ with entries $u_1, \dots, u_d$ or $v_1, \dots, v_d$, while we distinguish strictly between $\subset$ and $\subseteq$, \textit{\textit{i.e.}} if $J\subset I$ then $J\neq I$.
Moreover, for a univariate distribution function $F$ we define its inverse as $F^{-1}(u):=\inf\{ x\in\mathbb{R}\,|\, F(x)\ge u\}$, while we call a function $f:\mathbb{R}^d\rightarrow\mathbb{R}$ left-continuous if it is left-continuous in each component.

\begin{definition}
A function $Q: \mathbb I^d\rightarrow \mathbb I$ is a \textit{$d$-quasi-copula} if it satisfies the following properties:
\begin{itemize}
\item[(C1)]	boundary condition: $Q(u_1, \dots, u_i=0, \dots, u_d)=0$, for all $1\le i\le d$. 
\item[(C2)] uniform marginals: $Q(1, \dots, 1, u_i, 1, \dots, 1)=u_i$, for all $1\le i\le d$.
\item[(C3)] $Q$ is non-decreasing in each component.
\item[(C4)] $Q$ is Lipschitz continuous, \textit{i.e.} for all $\mathbf{u},\mathbf{v}\in\mathbb I^d$
		\begin{align*}
			|Q(\mathbf{u})-Q(\mathbf{v})|\le \sum_{i=1}^d |u_i-v_i| \,.
		\end{align*}
\end{itemize}
Moreover, $Q$ is a $d$-\textit{copula} if it satisfies in addition:
\begin{itemize}
\item[(C5)] $Q$ is $d$-increasing.
\end{itemize}
\end{definition}

The set of all $d$-quasi-copulas is denoted by $\mathcal{Q}^d$ and the set of all $d$-copulas by $\mathcal{C}^d$. 
Obviously, $\mathcal{C}^d \subset \mathcal Q^d$.
Moreover, we call $Q\in\mathcal{Q}^d\setminus\mathcal{C}^d$ a \textit{proper} quasi-copula.
In case the dimension $d$ is clear, we refer to a $d$-(quasi-)copula as a (quasi-)copula.

There exists a clear link between copulas and probability distributions. 
In fact, for $C\in\mathcal{C}^d$ and univariate distribution functions $F_1, \dots, F_d$,
\begin{align}\label{eq:sklar}
	F(\mathbf{x}):=C\big(F_1(x_1), \dots, F_d(x_d)\big)
\end{align}
defines a $d$-dimensional distribution function with marginals $F_1, \dots, F_d$. 
The celebrated theorem of \citet{sklar} tells us that the converse is also true, \textit{i.e.} given a $d$-dimensional distribution function $F$ with univariate marginals $F_1, \dots, F_d$, there exists a copula $C$ such that \eqref{eq:sklar} holds true.
We will call $C$ the copula corresponding to $F$.

Let $Q\in\mathcal{Q}^d$. 
We define its \textit{survival function} $\widehat{Q}: \mathbb I^d\rightarrow\mathbb I$ as follows:
\begin{align}\label{eq:surv-copula}
	\widehat{Q}(\mathbf{u}) := V_Q\Big(\bigtimes_{i=1}^d(u_i,1]\Big),
\end{align}
and denote by $\widehat{\mathcal{C}}^d:=\{\widehat{C}\,|\,C\in\mathcal{C}^d$\}.
A well-known result states that if $C\in\mathcal C^d$, then $\mathbf u\mapsto \widehat C(\mathbf 1-\mathbf u)$ is again a copula, namely the \textit{survival copula} of $C$, while there exists also a version of Sklar's theorem for survival copulas. 
In case $Q$ is a proper quasi-copula, then $\mathbf u\mapsto \widehat Q(\mathbf 1-\mathbf u)$ is not a quasi-copula in general; see \textit{e.g.} Example 2.5 in \citet{lux2016}.

Let us now define a partial order on $\mathcal{Q}^d$, and thus also on $\mathcal{C}^d$.
\begin{definition}
Let $Q_1, Q_2\in\mathcal{Q}^d$.
\begin{itemize}
\item[(i)] If $Q_1(\mathbf{u})\le Q_2(\mathbf{u})$ for all $\mathbf{u}\in\mathbb I^d$, then $Q_1$ is smaller than $Q_2$ in the \textit{lower orthant order}, denoted by $Q_1 \preceq_{LO} Q_2$.
\item[(ii)] If $\widehat{Q}_1(\mathbf{u})\le \widehat{Q}_2(\mathbf{u})$ for all $\mathbf{u}\in\mathbb I^d$, then $Q_1$ is smaller than $Q_2$ in the \textit{upper orthant order}, denoted by $Q_1 \preceq_{UO} Q_2$.
\end{itemize}
\end{definition}

The celebrated \FH bounds provide upper and lower bounds for all quasi-copulas with respect to the lower orthant order. 
Indeed, for $Q\in\mathcal Q^d$, we have that
\begin{align*}
	W_d(\mathbf{u}) :=\max\bigg\{\sum_{i=1}^d u_i -d+1,0\bigg\} \le Q(\mathbf u) \le \min\{u_1, \dots, u_d\} =: M_d(\mathbf{u}),	
\end{align*}
for all $\mathbf u\in\mathbb I^d$, which readily implies that $W_d\preceq_{LO} C\preceq_{LO} M_d$.
$W_d$ and $M_d$ are respectively called the lower and upper Fr\'echet--Hoeffding bounds.
Analogous results hold true for the upper orthant order and the survival functions, \textit{i.e.} we have that
\begin{align*}
	W_d(\mathbf 1 - \mathbf u) \le \widehat{C}(\mathbf{u}) \le M_d(\mathbf 1-\mathbf u), \quad \text{ for all }\mathbf{u}\in\mathbb I^d\,,\
\end{align*}
while an easy computation shows that $M_d(\mathbf 1-\cdot)=\widehat M_d(\cdot)$ for all $d\ge2$, while $W_d(\mathbf 1-\cdot)=\widehat W_d(\cdot)$ only for $d=2$.

The \FH bounds are derived under the assumption that the marginal distributions are fully known and the copula is fully unknown.
However, in several applications such as finance and insurance, partial information on the copula is available from market data.
Therefore, there has been intensive research in the last decade on improving the \FH bounds by adding partial information on the copula, see \textit{e.g.} \citet{Nel06}, \citet{tankov}, \citet{puccetti2016}, \citet{lux2016,Lux_Papapantoleon_2019}; see also \citet{Rueschendorf_2018} and reference therein.
Earlier results in this direction, using the language of optimal transport, appear in \citet{Rachev_Rueschendorf_1994}.
The following results from \citet[Sec. 3]{lux2016} describe \textit{improved} \FH bounds under the assumption that the copula is known in a subset of its domain, or that a functional of the copula is known.
Analogous statements for survival copulas appear in \cite[Appendix A]{lux2016}.

Let $\mathcal{S}\subseteq[0,1]^d$ be compact and $Q^*\in\mathcal{Q}^d$. 
Define the set
\begin{align*}
	\mathcal{Q}^{\mathcal{S}, Q^*} := \big\{Q\in\mathcal{Q}^d\, |\, Q(\mathbf{x})=Q^*(\mathbf{x}) \,\,\textit{for all}\,\,\mathbf{x}\in\mathcal{S}\big\} \,.
\end{align*}
Then, for all $Q\in\mathcal{Q}^{\mathcal{S},Q^*}$
\begin{align*}
	Q_L^{\mathcal{S},Q^*}\preceq_{LO} Q\preceq_{LO} Q_U^{\mathcal{S},Q^*} \,,
\end{align*}
where the \textit{improved \FH bounds} $Q_L^{\mathcal{S},Q^*}, Q_U^{\mathcal{S},Q^*}\in\mathcal{Q}^d$ and are provided by
\begin{align*}
	Q_L^{\mathcal{S},Q^*}(\mathbf{u})&=\max\Big\{0,\sum_{i=1}^d u_i -d+1, \max_{\mathbf{x}\in\mathcal{S}}\big\{Q^*(\mathbf{x})-\sum_{i=1}^d (x_i-u_i)^+ \big\}\Big\}, \\
	Q_U^{\mathcal{S},Q^*}(\mathbf{u})&=\min\Big\{u_1, \dots, u_d, \min_{\mathbf{x}\in\mathcal{S}}\big\{Q^*(\mathbf{x})+\sum_{i=1}^d (u_i-x_i)^+ \big\} \Big\} \,.
\end{align*}

\begin{remark}
A natural question is whether the bounds $Q_L^{\mathcal{S},Q^*}$ and $Q_U^{\mathcal{S},Q^*}$ are copulas or proper quasi-copulas. 
\citet{Nel06} showed that in the case of $\mathcal{S}$ being a singleton and for $d=2$ the lower and upper improved Fr\'echet-Hoeffding bounds are copulas using the concept of shuffles of $M_2$. 
This statement was generalized by \citet{tankov} and \citet{bernard}, still for $d=2$, under certain ``monotonicity'' conditions.
On the contrary, \citet{lux2016} showed that for $d>2$, the improved \FH bounds are copulas only in trivial cases and proper quasi-copulas otherwise.
Moreover, \citet{Bartl_Kupper_Lux_Papapantoleon_2017} showed that the improved \FH bounds are not pointwise sharp (or best-possible), even in $d=2$, if the aforementioned ``monotonicity'' conditions are violated.
\end{remark}

The next result provides improved \FH bounds in case the value of a functional of the copula is known. 
Examples of functionals could be the correlation or another measure of dependence (\textit{e.g.} Kendall's $\tau$ or Spearman's $\rho$), but also prices of multi-asset options in a mathematical finance context.
Let $\rho:\mathcal{Q}^d\rightarrow\mathbb{R}$ be non-decreasing with respect to the lower orthant order and continuous with respect to the pointwise convergence of quasi-copulas, and consider the set of quasi-copulas
\begin{align}\label{eq:sets-q-rho-theta}
	\mathcal{Q}^{\rho,\theta} := \big\{ Q\in\mathcal{Q}^d \,|\, \rho(Q)=\theta \big\} \,,
\end{align}
for $\theta\in[\rho(W_d),\rho(M_d)]$.
Then, for all $Q\in\mathcal{Q}^{\rho,\theta}$, holds
\begin{align*}
	Q_L^{\rho,\theta} \preceq_{LO} Q \preceq_{LO} Q_U^{\rho,\theta}\,,
\end{align*}
where the \textit{improved \FH bounds} $Q_L^{\rho,\theta}, Q_U^{\rho,\theta}\in \mathcal Q^d$ are provided by
\begin{align} \label{eq:iFH-func-L}
Q_L^{\rho,\theta}(\mathbf{u}) 
	:=&\begin{cases} 
		\rho^{-1}_+(\mathbf{u},\theta)\,,\quad \textit{if}\,\, \theta\in[\rho_+(\mathbf{u}, W_d(\mathbf{u})),\rho(M_d)], \\
		W_d(\mathbf{u})\,, \quad\quad \textit{otherwise}\,,
	\end{cases} \\ \label{eq:iFH-func-U}
Q_U^{\rho,\theta}(\mathbf{u})
	:=&\begin{cases} 
		\rho^{-1}_-(\mathbf{u},\theta)\,,\quad \textit{if}\,\, \theta\in[\rho(W_d),\rho_-(\mathbf{u},M_d(\mathbf{u}))], \\
		M_d(\mathbf{u})\,, \quad\quad \textit{otherwise}.
	\end{cases}
\end{align}
The improved \FH bounds $Q_L^{\rho,\theta}$ and $Q_U^{\rho,\theta}$ are actually the infimum and supremum over the set $\mathcal Q^{\rho,\theta}$.
Here we use the following notation: for $\mathbf{u}\in[0,1]^d$, let $r\in I_{\mathbf{u}}=[W_d(\mathbf{u}),M_d(\mathbf{u})]$ and $Q^*\in\mathcal{Q}^d$ with $Q^*(\mathbf{u})=r$, and define $Q_L^{\{\mathbf{u}\},r}:=Q_L^{\{\mathbf{u}\},Q^*}, Q_U^{\{\mathbf{u}\},r}:=Q_U^{\{\mathbf{u}\},Q^*}$ and 
\begin{align*}
	\rho_-(\mathbf{u},r) := \rho\big(Q_L^{\{\mathbf{u}\},r}\big) 
		\quad \text{ and } \quad
	\rho_+(\mathbf{u},r) := \rho\big(Q_U^{\{\mathbf{u}\},r}\big) \,.
\end{align*}
Then, for fixed $\mathbf{u}$, the maps $r\mapsto\rho_-(\mathbf{u},r)$ and $r\mapsto\rho_+(\mathbf{u},r)$ are non-decreasing and continuous. 
Hence, we can define their inverse mappings 
\begin{align*}
	\theta\mapsto \rho_-^{-1}(\mathbf{u},\theta)&:=\max\{r\in I_{\mathbf{u}} :\rho_-(\mathbf{u},r)=\theta\}, \\
	\theta\mapsto \rho_+^{-1}(\mathbf{u},\theta)&:=\min\{r\in I_{\mathbf{u}} :\rho_+(\mathbf{u},r)=\theta\},
\end{align*}
for all $\theta$ such that the sets are non-empty.
Analogous statements for non-increasing functionals are relegated to Appendix \ref{app:B}.


\section{Integration and stochastic dominance for quasi-copulas}
\label{sec:ISD-CQ}

This section provides results on the definition of integrals with respect to quasi-copulas and on stochastic dominance for quasi-copulas.
These results are largely taken from \citet[Sec. 5]{lux2016}, however we also provide a new representation of the integral with respect to a quasi-copula, as well as some useful results on stochastic dominance for quasi-copulas.
Stochastic dominance results for copulas were introduced in \citet{Rueschendorf_1980}, see also the book of \citet{mueller}, while analogous results for quasi-copulas in the 2-dimensional case appear in \citet{tankov}.
A special case of the results on stochastic dominance presented below appears already in \citet{Rueschendorf_2004}.

Let $(\Omega,\mathcal{F},\mathbb{P})$ be a probability space. 
Consider an $\mathbb{R}^d_+$-valued random vector $\mathbf{X}=(X_1, \dots, X_d)$ with distribution function $F$ and marginals $F_1, \dots, F_d$. 
Then, from Sklar's Theorem, we know there exists a copula $C\in\mathcal{C}^d$ such that 
\begin{align*}
	\mathbb{P}(X_1<x_1, \dots, X_d<x_d) = C\big(F_1(x_1), \dots, F_d(x_d)\big)
\end{align*}
and
\begin{align*}
	\mathbb{P}(X_1>x_1, \dots, X_d>x_d) = \widehat{C}\big(F_1(x_1), \dots, F_d(x_d)\big) \,.
\end{align*}
Hence, there exists an induced measure $\ud C\big(F_1(x_1), \dots, F(x_d)\big)$ on $\mathbb{R}^d_+$. 
Consider a function $f: \mathbb{R}^d_+\rightarrow\mathbb{R}$. 
In this section we focus on calculating $\mathbb{E}[f(\mathbf{X})]$ and its properties with respect to $C$. 
Assuming the marginals are given, we define the \textit{expectation operator} $\pi_f$ as follows
\begin{equation}
	\begin{aligned}\label{expectationOperator}
\pi_f(C) :=\mathbb{E}[f(\mathbf{X})]
	&= \int\limits_{\mathbb{R}^d_+} f(x_1, \dots, x_d) \,\ud C\big(F_1(x_1), \dots, F_d(x_d)\big) \\
	&= \int\limits_{[0,1]^d} f\big(F^{-1}_1(u_1), \dots, F^{-1}_d(u_d)\big) \,\ud C(u_1, \dots, u_d) \,.
	\end{aligned}
\end{equation}
However, if $Q$ is a proper quasi-copula then $\ud Q\big(F_1(x_1), \dots, F(x_d)\big)$ does not induce a measure anymore, because the $Q$-volume $V_Q$ is not necessarily positive.
The idea is now to switch the function we integrate against, \textit{i.e.} to perform a Fubini transformation. 
In order to do so, the function $f$ has to induce a measure. 
Therefore, we consider functions of the following type.

\begin{definition}
	\begin{enumerate}[label=(\roman*)]
		\item A function $f:\mathbb{R}^d_+\rightarrow\mathbb{R}$ is called \emph{$\Delta$-antitonic} if for every subset $I=\{i_1, \dots, i_n\}\subseteq\{1, \dots, d\}$ with $|I|\ge1$ and every hypercube $\bigtimes_{j=1}^n(a_j, b_j]\subset\mathbb{R}^n_+$
			\begin{align*}
				(-1)^n\Delta_{a_1,b_1}^{i_1} \circ \dots\circ\Delta_{a_n,b_n}^{i_n} f\ge0 \,.
			\end{align*}
		\item A function $f:\mathbb{R}^d_+\rightarrow\mathbb{R}$ is called \emph{$\Delta$-monotonic} if for every subset $I=\{i_1, \dots, i_n\}\subseteq\{1, \dots, d\}$ with $|I|\ge1$ and every hypercube $\bigtimes_{j=1}^n(a_j, b_j]\subset\mathbb{R}^n_+$
			\begin{align*}
				\Delta_{a_1,b_1}^{i_1} \circ \dots\circ\Delta_{a_n,b_n}^{i_n} f\ge0 \,.
			\end{align*}
	\end{enumerate}
\end{definition}

We will frequently deal with marginals of functions $f$ and quasi-copulas $Q$, therefore the following definition is useful.

\begin{definition}
	\begin{enumerate}[label=(\roman*)]
		\item Let $f:\mathbb{R}^d_+\rightarrow\mathbb{R}$. Then, for $I=\{i_1, \dots, i_n\}\subseteq\{1, \dots, d\}$, we define the \emph{$I$-margin} of $f$ as
		\begin{align*}
			f_I:\mathbb{R}^n_+\rightarrow\mathbb{R} \,,(x_{i_1}, \dots, x_{i_n})\mapsto f(x_1, \dots, x_d),\text{with }x_k=0 \text{ for } k\notin I\,.
		\end{align*}
		\item Let $Q\in\mathcal{Q}^d$. Then, for $I=\{i_1, \dots, i_n\}\subseteq\{1, \dots, d\}$, we define the \emph{$I$-margin} of $Q$ as
		\begin{align*}
			Q_I:[0,1]^n\rightarrow[0,1] \,,(u_{i_1}, \dots, u_{i_n})\mapsto Q(u_1, \dots, u_d),\text{with }u_k=1  \text{ for } k\notin I\,.
		\end{align*}
	\end{enumerate}
\end{definition}

According to well-known results, we can associate a measure to every left-continuous and $\Delta$-monotonic or $\Delta$-antitonic function $f:\mathbb{R}^d_+\rightarrow\mathbb{R}$ via
\begin{align}
	\mu_{f_I}(\emptyset) := 0 \ \ \text{ and } \ \	\mu_{f_I}\big( R \big) := V_{f_I}\big( R \big)\,,	
\end{align} 
for every hyperrectangle $R\subseteq \R^{|I|}$.
Then, we get that $\mu_{f_I}$ is a positive measure on $\mathbb{R}^{|I|}_+$ if $f$ is $\Delta$-monotonic, and that $(-1)^n\mu_{f_I}$ is a positive measure on $\mathbb{R}^{|I|}_+$ if $f$ is $\Delta$-antitonic. 
If $I=\{1, \dots, d\}$, then we write $\mu_f$ instead of $\mu_{f_I}$. 
In addition, we define $\mu_{f_\emptyset}:=\delta_0$, where $\delta$ denotes the Dirac measure.

\begin{remark}\thlabel{minusf}
Let $f:\mathbb{R}^d_+\rightarrow\mathbb{R}$	be a left-continuous function, such that $-f$ is either $\Delta$-antitonic or $\Delta$-monotonic.
Then, we have for $I=\{i_1, \dots, i_n\}\subseteq\{1, \dots, d\}$ with $|I|\ge1$ and every hypercube $\bigtimes_{j=1}^n(a_j, b_j]\subset\mathbb{R}^n_+$ that
\begin{align*}
	(-1)^n\Delta_{a_1,b_1}^{i_1} \circ \dots\circ\Delta_{a_n,b_n}^{i_n} f&\le0 \,, \quad\textit{if $\,-f$ is $\Delta$-antitonic, and} \\
	\Delta_{a_1,b_1}^{i_1} \circ \dots\circ\Delta_{a_n,b_n}^{i_n} f&\le0 \,, \quad\textit{if $\,-f$ is $\Delta$-monotonic.}
\end{align*}
Hence, $-\mu_{f_I}$ is a positive measure on $\mathbb{R}^{|I|}_+$ if $-f$ is $\Delta$-monotonic and $(-1)^{n+1}\mu_{f_I}$ is a positive measure on $\mathbb{R}^{|I|}_+$ if $-f$ is $\Delta$-antitonic.
\end{remark}

The following definitions show how the measure induced by the $I$-marginals of functions in conjunction with the $I$-marginals of copulas, can be used to define an integration operation. 
We define iteratively: 
\begin{equation}
	\begin{aligned}\label{operatorVarphi}
	\textrm{for }|I|=0 : \varphi_f^I(C):=&f(0, \dots, 0) \,,\\
	\textrm{for }|I|=1 : \varphi_f^I(C):=&\int\limits_{\mathbb{R}_+} f_{i_1} (x_{i_1}) \,\ud F_{i_1}(x_{i_1}) \,,\\ 
	\textrm{for }|I|=n\ge2 : \varphi_f^I(C):=&\int\limits_{\mathbb{R}^{|I|}_+} \widehat{C}_I \big(F_{i_1}(x_{i_1}), \dots,F_{i_n}(x_{i_n})\big) \,\ud\mu_{f_I}(x_{i_1}, \dots, x_{i_n}) \\
	&+ \sum_{J\subset I}(-1)^{n+1-|J|}\varphi_f^J(C) \,,
	\end{aligned}
\end{equation}
where $\widehat{C}_I$ denotes the survival function of the $I$-margin of $C$.
\citet[Prop. 5.3]{lux2016} proved that the operator $\varphi_f^{\{1,\dots,d\}}(C)$ defined above coincides with the expectation operator $\pi_f(C)$ in \eqref{expectationOperator} in case $f:\mathbb{R}^d\rightarrow\mathbb{R}$ is left-continuous, $\Delta$-antitonic or $\Delta$-monotonic and $C\in\mathcal{C}^d$.
However, the operator $\varphi_f^{\{1,\dots,d\}}(C)$ does not depend on $C$ being a copula, and can be also defined for quasi-copulas.
This motivates the following definition, which generalizes the expectation operator to quasi-copulas.

\begin{definition}
Let $f: \mathbb{R}^d\rightarrow\mathbb{R}$ be left-continuous, $\Delta$-antitonic or $\Delta$-monotonic and $d\ge2$.
Then, the \textit{expectation operator} is defined as follows $\pi_f:\mathcal{Q}^d\rightarrow\mathbb{R}, \,Q\mapsto\pi_f(Q)$\,, with
	\begin{align*}
		\pi_f(Q):=\varphi_f^{\{1, \dots, d\}}(Q) \,.
	\end{align*}
\end{definition}

\begin{remark}\label{rem:dual-operators}
Let $Q\in\mathcal{Q}^d$ and consider its survival function $\widehat Q$.
We define the dual to the operations $\varphi_f^I$ and $\pi_f$ as follows:
\begin{align*}
\widehat\varphi_f^I \big(\widehat Q\big) := \varphi_f^I \big(Q\big) 
	\quad \text{ and } \quad
\widehat{\pi}_f \big(\widehat{Q}\big) := \pi_f \big(Q\big) \,,
\end{align*}
since both operations actually only depend on the knowledge of $\widehat Q$ and not of $Q$ itself.
\end{remark}

\begin{remark}
Using that $V_{f_{i}}\big((0,x]\big)=f_{i}(x)-f_{i}(0)$, we can rewrite the case $|I|=\{i\}$ from \eqref{operatorVarphi} as follows 
\begin{align}
	\int\limits_{\mathbb{R}_+} f_{i} (x_i) \,\ud F_i(x_i) = \int\limits_{\mathbb{R}_+} \big(1-F_i(x_i)\big) \ud\mu_{f_{i}}(x_i) +f_{i}(0)\,. \label{FubiniFlipI=1}
\end{align}
Depending on the way the integrals are computed, this representation might be more useful. 
If we compute the one-dimensional integrals as in \eqref{operatorVarphi} instead of \eqref{FubiniFlipI=1}, then we do not need $f_{\{i\}}$ to induce a measure. 
Therefore, in \cite{lux2016} the authors define $\Delta$-antitonic and $\Delta$-monotonic in the sense that only $f_I, |I|\ge2$, has to induce a measure.
\end{remark}

The following result provides an alternative, simpler representation for the expectation operator $\pi_f(Q)$.

\begin{theorem}\thlabel{lemmaPhi}
Let $f: \mathbb{R}^d\rightarrow\mathbb{R}$ be left-continuous, $\Delta$-antitonic or $\Delta$-monotonic and $Q\in\mathcal{Q}^d$. 
Then, the following representation holds
\begin{align}
\pi_f(Q)
	&= \int\limits_{\mathbb{R}^d_+}\widehat{Q}\big(F_1(x_1), \dots, F_d(x_d)\big) \,\ud\mu_f(x_1, \dots, x_d) \nonumber\\
	&\quad+ \sum_{\substack{J\subset I\\ |J|=d-1}}\,\,\,\int\limits_{\mathbb{R}^{d-1}_+}\widehat{Q}_J\big(F_1(x_{i_1}), \dots, F_{i_{d-1}}(x_{i_{d-1}})\big) \,\ud\mu_{f_J}(x_{i_1}, \dots, x_{i_{d-1}}) \nonumber\\
	&\quad+ \dots +\sum_{i=1}^d\,\,\int\limits_{\mathbb{R}_+} f_{\{i\}}(x_i) \,\ud F_i(x_i) - (d-1) f(0, \dots, 0) \nonumber\\
	&= f(0, \dots, 0)+\sum_{n=1}^d \,\sum_{\substack{J\subseteq I \\ J=\{i_1, \dots, i_n\}}} \,\int\limits_{\mathbb{R}^{n}_+}\widehat{Q}_J\big(F_1(x_{i_1}), \dots, F_{i_n}(x_{i_{n}})\big) \,\ud\mu_{f_J}(x_{i_1}, \dots, x_{i_{n}}) \label{sumPhi}\,.
\end{align}
\end{theorem}

\begin{proof}
Without loss of generality we assume $I=\{1, \dots, d\}$. 
For $|I|=1$ the claim is given by \eqref{FubiniFlipI=1}. 
Now assume it holds all $n<d$ for some $d\in\mathbb{N}$. 
Define
\begin{align*}
	\alpha_{J}:=&\int\limits_{\mathbb{R}^{|J|}_+}\widehat{Q}_J\big(F_1(x_{i_1}), \dots, F_{i_n}(x_{i_{n}})\big) \,\ud\mu_{f_J}(x_{i_1}, \dots, x_{i_{n}}) \,, \,|J|\ge1 \,, \\
	\alpha_{\emptyset}:=& f(0, \dots, 0) \,.
\end{align*}
Then we deduce by \eqref{operatorVarphi} and the induction hypothesis
\begin{align}
	\varphi_f^I(Q) 	&= \alpha_I + \sum_{J\subset I}(-1)^{d+1-|J|}\varphi_f^J(Q) \nonumber \\
					&= \alpha_I+ \sum_{J\subset I}(-1)^{d+1-|J|} \sum_{J'\subseteq J} \alpha_{J'} \,. \label{proofAlpha}
\end{align}
Hence, we have to show that for every $J'\subset I$ the term $\alpha_{J'}$ appears exactly once in \eqref{proofAlpha} with positive sign. 
Consider $J'=\{j_1, \dots, j_k\}\subseteq J=\{i_1, \dots, i_n\}$. 
There are $\binom{d-k}{n-k}$ many $J\subset I$ with $J'\subseteq J$ because for $J\backslash J'$ we can choose $n-k$ elements out of $I\backslash J'$. 
We have
\begin{align*}
	\sum_{J\subset I}(-1)^{d+1-|J|} \sum_{J'\subseteq J} \alpha_{J'} &=\sum_{\substack{k=0 \\|J'|=k}}^n\sum_{n=k}^{d-1} (-1)^{d+1-n}\binom{d-k}{n-k} \alpha_{J'} \,.
\end{align*}
Further,
\begin{align*}
\sum_{n=k}^{d-1} (-1)^{d+1-n}\binom{d-k}{n-k}
   =\begin{cases} \sum\limits_{n=0}^{d-k-1}(-1)^{n+1} \binom{d-k}{n}, &\,\textrm{if $d-k$ is even,}\\ 
				   \sum\limits_{n=0}^{d-k-1}(-1)^n \binom{d-k}{n}, &\,\textrm{if $d-k$ is odd.}
	\end{cases}
\end{align*}
Since $\sum_{l=0}^m(-1)^l\binom{m}{l}=0$, $m\in\mathbb{N}$, we have $\sum_{n=k}^{d-1} (-1)^{d+1-n}\binom{d-k}{n-k}=1$ for both cases. 
This proves \eqref{sumPhi}. 
\end{proof}

Now we can show that the expectation operator $\pi_f$ is increasing or decreasing with respect to the lower and upper orthant order, depending on the properties of the function $f$.

\begin{proposition}\label{mono}
Let $Q_1, Q_2\in\mathcal{Q}^d$ and $f:\mathbb{R}_+^d\rightarrow\mathbb{R}$. 
Then
\begin{enumerate}[label=(\roman*)]
	\item for all $f$ $\Delta$-antitonic s.t. the integrals exist
			\begin{align*}
				Q_1\preceq_{LO} Q_2 \Longrightarrow \pi_f(Q_1)\le\pi_f(Q_2) \,,
			\end{align*}
	\item for all $f$ $\Delta$-monotonic s.t. the integrals exist
			\begin{align*}
				Q_1\preceq_{UO} Q_2 \Longrightarrow \pi_f(Q_1)\le\pi_f(Q_2) \,.
			\end{align*}
	\item for all $-f$ $\Delta$-antitonic s.t. the integrals exist
			\begin{align*}
				Q_1\preceq_{LO} Q_2 \Longrightarrow \pi_f(Q_1)\ge\pi_f(Q_2) \,,
			\end{align*}
	\item for all $-f$ $\Delta$-monotonic s.t. the integrals exist
			\begin{align*}
				Q_1\preceq_{UO} Q_2 \Longrightarrow \pi_f(Q_1)\ge\pi_f(Q_2) \,.
			\end{align*}
\end{enumerate}
\end{proposition}

\begin{proof}
The first two statements are \citet[Theorem 5.5]{lux2016}, while the next two are a direct consequence of them and Remark \ref{minusf}.
\end{proof}


\section{Copulas and arbitrage}
\label{sec:CA}

In this section, we apply the results on improved \FH bounds and on stochastic dominance for quasi-copulas to mathematical finance. 
We will first derive bounds for the arbitrage-free prices of certain classes of multi-asset derivatives. 
Then, we will formulate a necessary condition for the absence of arbitrage in markets where several multi-asset derivatives are traded simultaneously.


\subsection{Model and assumptions}

Let $(\Omega, \mathcal{F},\mathbb{P})$ be a probability space, where $\Omega=\mathbb R^d$. 
We consider the following financial market model: 
There exists one time period with initial time $t=0$ and final time $t=T<\infty$. 
Let $d\ge2$. 
There exist $d+1$ non-redundant primary assets denoted by $B, S^1, \dots, S^d$. 
We assume that their initial prices are known, \textit{i.e.} $(B_0, S^1_0, \dots, S^d_0)\in\mathbb{R}^{d+1}_+$. 
$B$ denotes the risk-free asset that earns the interest rate $r\ge0$ and, for the sake of simplicity, we set $B_T=1$, while $S^1_T, \dots, S^d_T$ are $\mathbb{R}_+$-valued random variables on the given probability space. 

This framework contains discrete- and continuous-time models as well, however we are only interested in the distribution of the random vector $S^1_T, \dots, S^d_T$.
This information is sufficient for the valuation of European vanilla derivatives, while the path dynamics are not necessary.
In the examples presented in the following section, the random vector is modeled according to the \citet{Black_Scholes_1973} model and an exponential L\'evy model; \textit{cf.} \citet{Eberlein01a}. 
The marginal distributions can be obtained by calibrating a parametric model to European option prices maturing at $T$ or by extracting the implied distributions from option prices using \textit{e.g.} the \citet{breeden} formula.

A probability measure $\mathbb{Q}$ on $(\Omega,\mathcal{F})$, equivalent to $\mathbb{P}$, that satisfies
\begin{align*}
	S^i_0 = B_0 \, \mathbb{E}_{\mathbb{Q}}\big[ S^i_T \big], \quad i=1, \dots, d \,\,,
\end{align*}
is called an equivalent martingale measure (EMM) for our financial market. 
The existence of an equivalent martingale measure and the absence of arbitrage have a well-known implication for the pricing of derivatives on $S^i_T$. 
Consider a European derivative of $S^i_T$ with payoff $H(S^i_T)$ at time $T$, where $H$ is a function such that $\mathbb{E}_{\mathbb{Q}}[H(S^i_T)]$ exists. 
Then, the arbitrage-free price is provided by
\begin{align*}
	H_0^i = B_0 \, \mathbb{E}_{\mathbb{Q}}\big[ H(S^i_T) \big], \quad i=1,\dots,d.
\end{align*}

We assume that the risk-neutral marginal distributions of each $S_T^i$ are known and unique for all $i=1,\dots,d$, \textit{i.e.} the univariate marginal distribution of $S_T^i$ under $\mathbb{Q}$ is equal for all EMMs $\mathbb{Q}$.
We further assume that these distributions are continuous, and denote them by $F_i$, \textit{i.e.}
\begin{align}\label{marginalsMartingaleMeasure}
	\mathbb{Q}\big( S^1_T\in\mathbb{R}_+, \dots, S^{i-1}_T\in\mathbb{R}_+, S^i_T\le x, S^{i+1}_T\in\mathbb{R}_+, \dots, S^d_T\in\mathbb{R}_+ \big) = F_i(x) \,, 
\end{align}
for all $i=1, \dots, d$. 
Let $\mathscr{P}$ denote the set of all EMMs for our financial market model, 
\textit{i.e.} $\mathscr{P}=\{\mathbb Q \,|\, \mathbb Q\sim\mathbb P, \mathbb Q \text{ MM}, F_i \ \text{continuous and satisfies \eqref{marginalsMartingaleMeasure}}\}$.
The assumption that the marginal distributions are known is not unrealistic, because their dynamics can be derived from market data; see \textit{e.g.} \citet{breeden}. 
This property implies, by the second Fundamental Theorem of Asset Pricing, that the prices of single-asset options are unique, and is referred to in the literature as \textit{static-completeness} of a financial market, see \textit{e.g.} \citet{Carr_Madan_2001}.
Let us stress that this does not imply $|\mathscr{P}|=1$, because the dependence structure of $S^1, \dots, S^d$ might not be uniquely determined. 

The financial market, beside options on the single assets $S^1, \dots, S^d$, consists also of a finite number of European multi-asset derivatives, denoted by $Z^1,\dots,Z^q$, for $q\in\mathbb{N}$. 
Their final payoffs at time $T$ are given by
\begin{align*}
	Z^i_T = z_i\big( S^1_T, \dots, S^d_T \big) \,,\quad i=1, \dots, q,
\end{align*}
where the payoff functions $z_i:\mathbb{R}^d_+\rightarrow\mathbb{R}_+$ (resp. their negation, \textit{i.e.} $-z_i$) are either $\Delta$-antitonic or $\Delta$-monotonic. 
We assume that $Z^1, \dots, Z^q$ are ``truly'' multi-asset derivatives, \textit{i.e.} they are written on at least two and up to $d$ of the risky assets.

\begin{definition}[Arbitrage-free price vector]
Let $(Z^1, \dots, Z^q)$ be a set of multi-asset derivatives as described above, for $q\in\mathbb{N}$. 
We call $\mathbf{p}=(p_1, \dots, p_q)\in\mathbb{R}_+^q$ an \textit{arbitrage-free price vector} for $(Z^1, \dots, Z^q)$ if there exists a measure $\mathbb{Q}\in\mathscr{P}$ such that
\begin{align*}
	p_k = B_0 \, \mathbb{E}_{\mathbb{Q}} \big[ Z^k_T \big] \,,\quad\textit{for all }k=1, \dots, q.
\end{align*}
We denote the set of all arbitrage-free price vectors for $(Z^1, \dots, Z^q)$ by $\Pi(Z^1, \dots, Z^q)$. 
This set is described by
\begin{align*}
\Pi(Z^1, \dots, Z^q) 
	= \Big\{ \Big( B_0 \, \mathbb{E}_{\mathbb{Q}}\big[ Z^1_T \big], \dots, B_0 \, \mathbb{E}_{\mathbb{Q}}\big[ Z^q_T \big] \Big) \,\Big|\, \mathbb{Q}\in\mathscr{P}\textit{ and } \mathbb{E}_{\mathbb{Q}}\big[Z^k_T\big]<\infty\,, k=1, \dots, q \Big\} \,.
\end{align*}
\end{definition}


\subsection{Copulas and arbitrage-free price vectors}

In this sub-section, we study the relation between copulas and the set of arbitrage-free price vectors $\Pi(Z^1, \dots, Z^q)$.
The first result appears already in \citet[Corollary 3]{Tav15}. 

\begin{proposition}\label{prop:c-p-bijection}
In the multi-asset financial market model described above, there is a bijection between $\mathscr{P}$ and $\mathcal{C}^d$.
\end{proposition}

\begin{proof}
This is a reformulation of Sklar's Theorem, \textit{cf.} \citet{sklar}, using the language of mathematical finance.
\end{proof}

This bijection allows us to express the arbitrage-free price of a derivative $Z_T^i$, and therefore also expectations of the form $\mathbb{E}_{\mathbb{Q}}[Z_T^i]$ for $\mathbb Q\in \mathscr P$, in terms of the associated copula $C_{\mathbb{Q}}$.
That is,
\begin{align}\label{integraldC}
\mathbb{E}_{\mathbb{Q}} \big[Z_T^i\big]
	&= \mathbb{E}_{\mathbb{Q}} \big[z_i(S_T^1, \dots, S_T^d)\big] \nonumber \\
	&= \int\limits_{\mathbb{R}_+^d}z_i(x_1, \dots, x_d) \,\ud F_{\mathbb{Q}}(\mathbf{x}) \nonumber \\
	&= \int\limits_{[0,1]^d}z_i\big(F^{-1}_1(u_1), \dots, F^{-1}_d(u_d)\big) \,\ud C_{\mathbb{Q}}(\mathbf{u}) \,.
\end{align}
We denote the expectation under the measure associated with a copula $C$ by $\mathbb{E}_C$. 
The bijection between the set of equivalent martingale measures and the set of copulas in Proposition \ref{prop:c-p-bijection} allows now to describe the set of arbitrage-free price vectors in terms of copulas, \textit{i.e.}
\begin{align}\label{eq:AFPV-copula}
\Pi(Z^1, \dots, Z^q) 
	= \Big\{ \Big( B_0 \, \mathbb{E}_C\big[Z_T^1\big], \dots, B_0 \, \mathbb{E}_C\big[Z_T^1\big] \Big) \,\Big|\, C\in\mathcal{C}^d\textit{ and } \mathbb{E}_{C}\big[Z^k_T\big]<\infty\,, k=1, \dots, q \Big\} \,.
\end{align}
Finally, recall the definition of the expectation operator $\pi_f$ from the previous section.
Using \eqref{expectationOperator} and \eqref{integraldC} we get that $\pi_{z_k}(C)=\mathbb{E}_C[Z_T^k]$ for $k=1, \dots,q$. 
Hence, for the multi-asset derivatives $Z^1, \dots, Z^q$ we define the following pricing rule between the set of copulas and the set of arbitrage-free price vectors $\Pi(Z^1, \dots, Z^q)$,
\begin{align*}
	\varrho:\mathcal{C}^d\rightarrow\mathbb{R}^q_+ \,, C\mapsto \varrho(C):=\big(B_0\, \pi_{z_1}(C), \dots, B_0\,\pi_{z_q}(C)\big) \,.
\end{align*}
Consequently, we can prove the following equivalence result.

\begin{proposition}\label{priceInPi}
Let $\mathbf{p}\in\mathbb{R}_+^q$. 
Then
\begin{align*}
	\mathbf{p}\in\Pi(Z^1, \dots, Z^q) \, \Longleftrightarrow \, \exists \,C\in\mathcal{C}^d \textit{ such that } \varrho(C)=\mathbf{p}\,.
\end{align*}
\end{proposition}

\begin{proof}
The equivalence follows immediately from the definition of the pricing rule together with \eqref{expectationOperator}, \eqref{integraldC} and  \eqref{eq:AFPV-copula}.
\end{proof}

\begin{remark}\thlabel{priceInPiSurvival}
Using the definition of the dual operator $\widehat \pi_f$, see Remark \ref{rem:dual-operators}, the previous result carries over analogously to the set of survival copulas $\widehat{\mathcal{C}}^d$, \textit{i.e.}
\begin{align*}
	\widehat{\varrho}:\widehat{\mathcal{C}}^d\rightarrow\mathbb{R}^q_+ \,, 
	\widehat{C}\mapsto \widehat{\varrho}(\widehat{C}):=\big(B_0\, \widehat{\pi}_{z_1}(\widehat{C}), \dots, B_0\,\widehat{\pi}_{z_q}(\widehat{C})\big)
\end{align*}
and
\begin{align*}
	\mathbf{p}\in\Pi(Z^1, \dots, Z^q) \Longleftrightarrow \exists \,\widehat{C}\in\widehat{\mathcal{C}}^d \textit{ such that } \widehat{\varrho}(\widehat{C})=\mathbf{p}\,.
\end{align*}
\end{remark}


\subsection{Bounds for the arbitrage-free price of a single multi-asset derivative}

We have assumed that the payoff functions $z_i:\R^d_+\to\R_+$, resp. their negations $-z_i$, are either $\Delta$-antitonic or $\Delta$-monotonic.
Therefore, we get from Proposition \ref{mono} that $\pi_{z_i}$ is non-decreasing, resp. non-increasing, with respect to the lower or upper orthant order. 
Hence, we can use the Fr\'echet--Hoeffding bounds and the parametrization of arbitrage-free price vectors in terms of copulas in order to derive arbitrage-free bounds for the set $\Pi(Z^i)$ for each multi-asset derivative in the market. 
Moreover, assume there exists additional information about the copulas, \textit{i.e.} consider a constrained set $\mathcal{C}^*\subseteq\mathcal{C}^d$ such as $\mathcal{C}^{\mathcal{S},C^*}$ or $\mathcal{C}^{\rho,\theta}$. 
Then, we also have a constrained set of arbitrage-free prices, \textit{i.e.}
\begin{align*}
	\Pi^*(Z^i) = \big\{B_0 \pi_{z_i}(C) \,|\, C\in\mathcal{C}^*\big\} \subseteq \Pi(Z^i) \,.
\end{align*}
In other words, the improved Fr\'echet--Hoeffding bounds allow us to tighten the range of arbitrage-free prices for the derivative $Z^i$. 
This concept works analogously for the set of survival functions, \textit{i.e.} for $\widehat{\mathcal{C}}^*\subset\widehat{\mathcal{C}}^d$.

\begin{corollary}
Let $Z$ be a multi-asset derivative in the financial market described above with payoff function $z$.
\begin{enumerate}[label=(\roman*)]
\item Let $z$ be $\Delta$-antitonic and $Q_L^*, Q_U^*$ be the lower and upper bound for some constrained set $\mathcal{C}^*\subseteq\mathcal{C}^d$.
	  Then, for all $C\in\mathcal{C}^*$ holds
		\begin{align*}
			\pi_z(W_d) \le \pi_z(Q_L^*) \le \pi_z(C) \le \pi_z(Q_U^*) \le \pi_z(M_d) \,.
		\end{align*}
\item Let $z$ be $\Delta$-monotonic and $\widehat{Q}_L^*, \widehat{Q}_U^*$ be the lower and upper bound for some constrained set $\widehat{\mathcal{C}}^*\subseteq\widehat{\mathcal{C}}^d$. 
	  Then, for all $\widehat{C}\in\widehat{\mathcal{C}}^*$ holds
		\begin{align*}
			\widehat{\pi}_z(\overline{W}_d) \le \widehat{\pi}_z(\widehat{Q}_L^*) \le \widehat{\pi}_z(\widehat{C}) \le \widehat{\pi}_z(\widehat{Q}_U^*) \le \widehat{\pi}_z(\overline{M}_d)=\pi_z(M_d)\,,
		\end{align*}
		where $\overline{W}_d(\mathbf{u})=W_d(\mathbf 1-\mathbf u)$ and $\overline{M}_d(\mathbf{u})=M_d(\mathbf 1-\mathbf u)$.
\end{enumerate}
\end{corollary}

\begin{proof}
These claims follow directly from the ordering of the bounds, the monotonicity results in Proposition \ref{mono}, and their analogues for survival functions.
\end{proof}

\begin{remark}
The inequalities above change direction if $-z$ is either $\Delta$-antitonic or $\Delta$-monotonic.	
\end{remark}

\begin{remark}
\citet[Section 6]{lux2016} provide conditions such that the improved option price bounds are sharp, in the sense that $\inf \Pi^*(Z) =\pi_z(Q^*_L)$ and $\sup \Pi^*(Z)=\pi_z(Q^*_U)$ respectively.
Depending on the payoff function $z$ the computation of the improved option price bounds can be quite complicated. 
\citet{RR01}, \citet{tankov} and \citet{lux2016} present several derivatives for which the integrals can be enormously simplified. 
\end{remark}


\subsection{A necessary condition for the absence of arbitrage in the presence of several multi-asset derivatives}

In this subsection, we assume there exist several multi-asset derivatives $Z^1, \dots, Z^q$ in the financial market, and consider a price vector $\mathbf{p}=(p_1, \dots, p_q)\in\mathbb{R}_+^q$ for them. 
Our goal is to check whether $\mathbf{p}$ is an arbitrage-free price vector or not, \textit{i.e.} whether $\mathbf{p}\in\Pi(Z^1, \dots, Z^q)$. 
In fact, we will derive a necessary condition for $\mathbf{p}$ to be an arbitrage-free price vector. 

Consider the following constrained sets of copulas
\begin{align*}
	\mathcal{C}^{\pi_k,p_k} := \big\{ C\in\mathcal{C}^d\,|\,B_0\,\pi_{z_k}(C)=p_k \big\} \,,\,k=1, \dots, q\,,
\end{align*}
which are sets of the form \eqref{eq:sets-q-rho-theta}.
This set contains all copulas that are compatible with the observed market price $p_k$ for the multi-asset derivative $Z^k$, for each $k=1,\dots,q$.
Clearly, $\mathcal{C}^{\pi_k,p_k}\neq\emptyset$ if and only if $p_k\in\Pi(Z^k)$ by Proposition \ref{priceInPi}. 
Analogously we define the set of survival functions
\begin{align*}
	\widehat{\mathcal{C}}^{\pi_k,p_k} := \big\{ \widehat{C}\in\widehat{\mathcal{C}}^d\,|\,B_0\,\widehat{\pi}_{z_k}(\widehat{C})=p_k \big\} \,, \,k=1, \dots, q\,.
\end{align*}

The next result shows that $\mathbf p$ is an arbitrage-free price vector for $(Z^1, \dots, Z^q)$ if and only if there exists a copula $C$ that reproduces the market prices of the derivatives $(Z^1, \dots, Z^q)$.
This copula will then necessarily belong to one of the sets ${\mathcal{C}}^{\pi_k,p_k}$ or $\widehat{\mathcal{C}}^{\pi_k,p_k}$, for $k=1,\dots,q$.

\begin{proposition}\label{prop:arb-cop-equiv}
Let $\mathbf{p}\in\mathbb{R}_+^q$. 
We have the following equivalences:  if $z_1,\dots,z_q$ are $\Delta$-antitonic, then
\begin{align*}
\mathbf{p}\in\Pi(Z^1, \dots, Z^q) 	&\Longleftrightarrow \bigcap_{k=1}^q \mathcal{C}^{\pi_k,p_k} \neq\emptyset
\intertext{while, if $z_1,\dots,z_q$ are $\Delta$-monotonic, then}
\mathbf{p}\in\Pi(Z^1, \dots, Z^q) 	&\Longleftrightarrow \bigcap_{k=1}^q \widehat{\mathcal{C}}^{\pi_k,p_k} \neq\emptyset \,. 
\end{align*}
\end{proposition}

\begin{proof}
Let $\mathbf{p}\in\Pi(Z^1, \dots, Z^q)$ and $z_1,\dots,z_q$ be $\Delta$-antitonic.
Then, there exists a $d$-copula $C\in\mathcal C^d$ such that $\varrho(C)=\mathbf p$ hence, for every $k=1,\dots,q$, there exists a $C\in\mathcal{C}^{\pi_k,p_k}$ such that $\varrho^k(C)=p_k$.
This readily implies that 
$$\bigcap_{k=1}^q \mathcal{C}^{\pi_k,p_k} \neq\emptyset.$$
Using the same arguments in the opposite direction allows to prove the equivalence.
The case for $\Delta$-monotonic functions is completely analogous.
\end{proof}

\begin{remark}
The previous result implies that the set of arbitrage-free price vectors for $Z^1, \dots, Z^q$ is a subset of the Cartesian product of the sets of arbitrage-free price vectors for each $Z^i$, \textit{i.e.}
\[
	\Pi(Z^1, \dots, Z^q) \subseteq \Pi(Z^1) \times \cdots \times \Pi(Z^q).
\]
In other words, we can have derivatives that are priced within their own no-arbitrage bounds, however when they are considered together an arbitrage opportunity may arise.
An example in this direction will be presented in the following section.
\end{remark}

The idea now is to find pointwise upper and lower bounds for the sets of copulas $\mathcal{C}^{\pi_k,p_k}$ and $\widehat{\mathcal{C}}^{\pi_k,p_k}, k=1, \dots,q$, and here the improved \FH bounds play a crucial role.
Let us define
\begin{align}
\overline Q^k_{\mathbf{p}}(\mathbf{u})
	&=\begin{cases} 
		Q_U^{\pi_k,p_k}(\mathbf{u}), \quad \textrm{if $z_k$ is $\Delta$-antitonic},\\ 
		\widehat{Q}_U^{\pi_k,p_k}(\mathbf{u}), \quad \textrm{if $z_k$ is $\Delta$-monotonic},
	\end{cases} \label{eq:iFH-upp}  \\ \label{eq:iFH-dow}
\underline Q^k_{\mathbf{p}}(\mathbf{u})
	&=\begin{cases} 
		Q_L^{\pi_k,p_k}(\mathbf{u}), \quad \textrm{if $z_k$ is $\Delta$-antitonic},\\  
		\widehat{Q}_L^{\pi_k,p_k}(\mathbf{u}), \quad\textrm{if $z_k$ is $\Delta$-monotonic},
	\end{cases}
\end{align}
where $Q_L^{\pi_k,p_k}, Q_U^{\pi_k,p_k}, \widehat{Q}_L^{\pi_k,p_k}, \widehat{Q}_U^{\pi_k,p_k}$ are defined as in \eqref{eq:iFH-func-L}--\eqref{eq:iFH-func-U}; see also \citet[Prop. A.2]{lux2016}. 
Moreover, we define
\begin{equation}\label{ApBp}
	\overline Q_{\mathbf{p}}(\mathbf{u}) := \min\big\{ \overline Q^k_{\mathbf p}(\mathbf{u})\,|\,k=1, \dots, q \big\}
		\quad \text{and} \quad
	\underline Q_{\mathbf{p}}(\mathbf{u}) := \max\big\{ \underline Q^k_{\mathbf p}(\mathbf{u})\,|\,k=1, \dots, q \big\}.
\end{equation}

Let us recall that $\mathcal{C}^{\pi_k,p_k}$ is the set of copulas that are compatible with the observed market prices of the multi-asset derivatives $Z^1,\dots,Z^q$.
Proposition \ref{prop:arb-cop-equiv} states that the market is free of arbitrage if and only if there exists a copula in the intersection of the sets $\mathcal{C}^{\pi_k,p_k}$ that is compatible with the observed market prices for these derivatives. 
The bounds $\overline Q^k_{\mathbf{p}}$ and $\underline Q^k_{\mathbf{p}}$ are the pointwise upper and lower bounds for the set $\mathcal{C}^{\pi_k,p_k}$, for each $k=1, \dots, q$, and dictate the maximal and minimal value that a copula compatible with the observed market prices can take.
Moreover, $\overline Q_{\mathbf{p}}$ and $\underline Q_{\mathbf{p}}$ are the minimal upper bound and the maximal lower bound over all $k=1,\dots,q$. 
The next Theorem is the main result of this section, and provides a necessary condition for the absence of arbitrage in a financial market in the presence of several multi-asset derivatives; this generalizes \citet[Proposition 9]{Tav15} to the $d$-dimensional case.
Indeed, if the market is free of arbitrage, then a copula will exist in the intersection of the sets $\mathcal{C}^{\pi_k,p_k}$ and thus $\overline Q_{\mathbf{p}} \ge \underline Q_{\mathbf{p}}$ pointwise.

\begin{theorem}\label{thm:arbitrage}
Let $\mathbf{p}\in\mathbb{R}_+^q$ and $z_1,\dots,z_q$ be either $\Delta$-antitonic or $\Delta$-monotonic. 
Then, in the financial market described above, with several multi-asset derivatives $Z^1, \dots, Z^q$ traded simultaneously, we have
\begin{align}\label{eq:arbitrage-implied}
	\mathbf{p}\in\Pi(Z^1, \dots, Z^q) \Longrightarrow \underline Q_{\mathbf{p}}(\mathbf{u})\le \overline Q_{\mathbf{p}}(\mathbf{u}) \quad\textit{for all }\mathbf{u}\in[0,1]^d\,.
\end{align}
\end{theorem}

\begin{proof}
Let $z_1,\dots,z_q$ be $\Delta$-antitonic. 
Assume there exists a $\mathbf{u}^*\in[0,1]^d$ such that $\underline Q_{\mathbf{p}}(\mathbf{u}^*)> \overline Q_{\mathbf{p}}(\mathbf{u}^*)$. 
By construction of $\overline Q_{\mathbf{p}}$ and $\underline Q_{\mathbf{p}}$, the minimum and maximum are always attained. 
Denote by $k_A, k_B\in\{1, \dots, q\}$ the indices for which the minimum and maximum are attained in \eqref{ApBp}. 
Then we have that $k_A\neq k_B$, because otherwise
\begin{align*}
\inf\big\{ C(\mathbf{u^*})\,|\,C\in\mathcal{C}^{\pi_{k_A},p_{k_A}} \big\} 
	=\underline Q_{\mathbf{p}}^{k_A}(\mathbf{u^*}) 
	&= \underline Q_{\mathbf{p}}(\mathbf{u}^*) \\
	&> \overline Q_{\mathbf{p}}(\mathbf{u}^*)
	=\overline Q_{\mathbf{p}}^{k_A}(\mathbf{u^*})
	=\sup\big\{ C(\mathbf{u^*})\,|\,C\in\mathcal{C}^{\pi_{k_A},p_{k_A}} \big\} \,.
\end{align*}
Hence, we get that
\begin{align*}
\underline Q_{\mathbf{p}}(\mathbf{u}^*) 
	= \inf\big\{ C(\mathbf{u^*})\,|\,C\in\mathcal{C}^{\pi_{k_B},p_{k_B}} \big\} 
	> \sup\{C(\mathbf{u^*})\,|\,C\in\mathcal{C}^{\pi_{k_A},p_{k_A}}\} 
	= \overline Q_{\mathbf{p}}(\mathbf{u}^*)\,,
\end{align*}
which readily implies that $\mathcal{C}^{\pi_{k_A},p_{k_A}}\cap\mathcal{C}^{\pi_{k_B},p_{k_B}}=\emptyset$. 
Therefore, we also get that
\begin{align*}
\bigcap_{k=1}^q \mathcal{C}^{\pi_k,p_k}\subseteq\Big( \mathcal{C}^{\pi_{k_A},p_{k_A}}\cap\mathcal{C}^{\pi_{k_B},p_{k_B}} \Big)=\emptyset \,,
\end{align*}
which is equivalent to $\mathbf{p}\notin\Pi(Z^1, \dots, Z^q)$ by Proposition \ref{prop:arb-cop-equiv}. 
The proof for $\Delta$-monotonic functions $z_1,\dots,z_q$ and $\widehat{\mathcal{C}}^{\pi_k,p_k}$ works completely analogously.
\end{proof}

We have assumed so far that there exist $S^1,\dots,S^d$ underlying assets in the financial market and that all multi-asset derivatives $Z^1,\dots,Z^q$ depend on all $d$ assets.
This is however not very realistic, as there might well exist derivatives that depend on some, but not all, of the underlying assets.
The next result treats exactly that scenario, making use of the results on $I$-margins of copulas.

Assume there exist $Z^1,\dots,Z^q$ multi-asset derivatives in the financial market, and that each derivative $Z^k$ depends on $d_k$ of the underlying assets with $2\le d_k\le d$.
That is, each $Z^k$ depends on $(S^{i_1}, \dots, S^{i_{d_k}})$ with  $I^k=\{i_1, \dots, i_{d_k}\}\subseteq\{1, \dots, d\}$ and $k=1,\dots,q$.
Let us define $I^*:=\bigcap_{k=1}^q I^k$ and $d^*:=|I^*|$. 
Moreover, we assume that $d^*\ge2$, \textit{i.e.} all multi-asset derivatives share at least two common underlying assets.

Let us now update the definition of the constrained set of copulas $\mathcal{C}^{\pi_k,p_k}$ as follows:
\begin{align*}
	\mathcal{C}^{\pi_k,p_k} := \big\{ C\in\mathcal{C}^d\,|\,B_0\,\pi_{z_k}(C_{I^k})=p_k \big\} \,,\,k=1, \dots, q\,;
\end{align*}
this coincides with the previous definition in case all derivatives depend on all $d$ assets.
Moreover, let us also define the following constrained set of copulas, that projects everything on the space of the common underlying assets:
\begin{align*}
	\mathcal{C}^{\pi_k,p_k}_{I^*} := \big\{ C_{I^*}\in\mathcal{C}^{d^*} \,|\, C \in \mathcal{C}^{\pi_k,p_k} \big\} \,,\,k=1, \dots, q\,.
\end{align*}
We define now the upper and lower improved \FH bounds for the set $\mathcal{C}^{\pi_k,p_k}_{I^*}$, denoted by $\overline Q^{k,*}_{\mathbf{p}}$ and $\underline Q^{,*}_{\mathbf{p}}$ completely analogously to \eqref{eq:iFH-upp} and \eqref{eq:iFH-dow}, and also define
\begin{equation}\label{ApBp-star}
	\overline Q_{\mathbf{p}}^*(\mathbf{u}) := \min\big\{ \overline Q^{,*}_{\mathbf p}(\mathbf{u})\,|\,k=1, \dots, q \big\}
		\quad \text{and} \quad
	\underline Q_{\mathbf{p}}^*(\mathbf{u}) := \max\big\{ \underline Q^{k,*}_{\mathbf p}(\mathbf{u})\,|\,k=1, \dots, q \big\},
\end{equation}
as in \eqref{ApBp}.
Then, we have the following necessary condition for the absence of arbitrage in this financial market.

\begin{theorem}\label{thm:arbitrage-2}
Let $\mathbf{p}\in\mathbb{R}_+^q$ and $z_1,\dots,z_q$ be either $\Delta$-antitonic or $\Delta$-monotonic. 
Then, in the financial market described above, with several multi-asset derivatives $Z^1, \dots, Z^q$ traded simultaneously, we have
\begin{align}\label{eq:arbitrage-implied-2}
	\mathbf{p}\in\Pi(Z^1, \dots, Z^q) \Longrightarrow \underline Q_{\mathbf{p}}^*(\mathbf{u})\le \overline Q_{\mathbf{p}}^*(\mathbf{u}) \quad\textit{for all }\mathbf{u}\in[0,1]^{d^*}\,.
\end{align}
\end{theorem}

\begin{proof}
The idea is again that for $\mathbf{p}\in\Pi(Z^1, \dots, Z^q)$ there must exist a $d^*$-copula $C$ with $C\in\bigcap_{k=1}^q \mathcal{C}^{\pi_k,p_k}_{I^*}$. 
The proof is then completely analogous to the proof of Theorem \ref{thm:arbitrage}, and thus omitted for the sake of brevity.
\end{proof}

The intuition behind the last two results is that whenever the inequalities in \eqref{eq:arbitrage-implied} and \eqref{eq:arbitrage-implied-2} are violated for some $\mathbf u \in[0,1]^d$, then there does not exist a copula that can describe the prices of all derivatives $Z^1,\dots,Z^q$.
Hence, this set of prices is not jointly arbitrage-free.
Therefore, following \citet{Tav15}, we can also express the arbitrage detection problem as a minimization problem. 
Indeed, let us consider,
\begin{align*}
\mathscr{O}: \quad \min_{\mathbf{u}\in[0,1]^d} \big\{ \overline Q_{\mathbf{p}}(\mathbf{u})-\underline Q_{\mathbf{p}}(\mathbf{u}) \big\} \,.
\end{align*}
The objective function $\mathbf{u}\mapsto \overline Q_{\mathbf{p}}(\mathbf{u})-\underline Q_{\mathbf{p}}(\mathbf{u})$ takes values in $[-1,1]$ and the minimization is realized over a compact set. 
Hence, there exists a (possibly not unique) minimum, say $\mathbf{u}^*\in[0,1]^d$. 
The idea now is that if $\overline Q_{\mathbf{p}}(\mathbf{u^*})-\underline Q_{\mathbf{p}}(\mathbf{u^*})<0$, then $\mathbf{p}$ is not free of arbitrage. 
Note that the opposite result would not necessarily imply $\mathbf{p}$ being arbitrage-free, since Theorems \ref{thm:arbitrage} and \ref{thm:arbitrage-2} provide only a necessary condition. 
Nevertheless $\mathscr{O}$ might detect an arbitrage which is not obvious in the first place. 
In fact, it is possible that $\mathbf{p}=(p_1, \dots, p_q)$ is not free of arbitrage although all $p_i$'s lie within the arbitrage-free bounds computed from the Fr\'echet-Hoeffding bounds.
In summary, we have the following result:
\[
\min_{u\in\mathbb I^d} \big\{ \overline Q_{\mathbf p} (\mathbf u) - \underline Q_{\mathbf p} (\mathbf u) \big\}
  = \begin{cases}
  \ge 0, & \text{no decision}, \\
  < 0, & \pi\notin\Pi.
  \end{cases}
\]


\section{Applications}
\label{sec:appl}

In this section, we present some applications of the previous results in the computation of bounds for arbitrage-free prices and in the detection of arbitrage opportunities.
We are particularly interested in the case where the prices of each multi-asset derivative lie within their respective no-arbitrage bounds, yet an arbitrage arises when they are considered jointly. 
In the numerical experiments we will use both artificial data as well as real market data.

The framework for the applications and the numerical examples presented below is summarized in the following bullet points:
\begin{itemize}
\item We consider a financial market as described above with final time $T=1$.
\item We assume, for simplicity, that the interest rate is zero, \textit{i.e.} $B_t=1\,, t\in[0,1]$\,.
\item There exist three risky assets $S^1, S^2, S^3$ ($d=3$) with known marginals distributions $F_1, F_2, F_3$ at $t=1$ but unknown dependence structure.
\item \textbf{Setting 1:} The marginals of $(S^1_1,S^2_1,S^3_1)$ are log-normally distributed, \textit{i.e.}
	\begin{align*}
		S^i_1=S_0^i \exp\Big(\sigma_i W_1^{(i)} -\frac{\sigma_i^2}{2}\Big)\,, \quad i=1,2,3
	\end{align*}
	where $W^{(i)}$ are standard Brownian motions, while the initial values and parameters are presented in Table \ref{table:data}.
\item \textbf{Setting 2:} The marginals of $(S^1_1,S^2_1,S^3_1)$ follow an exponential L\'evy model, \textit{i.e.}
      \begin{align*}
            S^i_1=S_0^i \exp\big( X_1^{(i)} \big)\,, \quad i=1,2,3
      \end{align*}
      where $X^{(i)}$ follows the normal inverse Gaussian distribution, \textit{i.e.} $X^{(i)}\sim \text{NIG}(\alpha_i,\beta_i,\delta_i,\mu_i)$. 
      The model parameters presented in the Table \ref{table:data} stem from calibrating the model to the option prices of the OBX, Coca Cola and StatoilHydro stocks, see \citet{Saebo_2009}, while the parameter $\mu$ is determined by the martingale condition.
      \begin{table}[h!]
      \begin{minipage}[b]{0.45\linewidth}
            \centering
            \begin{tabular}{c || c | c | c}
                  $i$ & 1 & 2 & 3 \\
                  \hline
                  $S^i_0$ & 8 & 10 & 12 \\
                  $\sigma_i$ & 1.5 & 1 & 0.5
            \end{tabular}
      \end{minipage}
      \begin{minipage}[b]{0.45\linewidth}
            \centering
            \begin{tabular}{c || c | c | c}
                  $i$ & 1 & 2 & 3 \\
                  \hline
                  $S^i_0$    &     100 &      100 & 100 \\
                  $\alpha_i$ &  8.9932 &  26.4502 & 9.7278 \\
                  $\beta_i$  & -4.5176 & -17.3990 & -3.2261 \\
                  $\delta_i$ &  1.1528 &   0.8872 & 1.1524
            \end{tabular}
      \end{minipage}
      \caption{\label{table:data}Artificial parameters for the log-normal model (left) and parameters from calibrated real market data for the NIG model (right).}
      \end{table}
\item There exist two multi-asset derivatives $Z^1, Z^2$ ($q=2$), with payoff functions $z_1, z_2$ such that $z_1$ and $-z_2$ are $\Delta$-monotonic.
\item The payoff functions of $Z^1$ and $Z^2$ are provided by
\begin{align*}
	z_1(\mathbf{x}) = \big( \min\{x_1, x_2, x_3\}-K_1 \big)^+ \quad \text{ and } \quad
	z_2(\mathbf{x}) = \big( K_2-\min\{x_1, x_2, x_3\} \big)^+ 
\end{align*}
for $K_1, K_2\in\mathbb{R}_+$, \textit{i.e.} a call and a put option on the minimum of three assets.
\end{itemize}


\subsection{Bounds for arbitrage-free prices within the two sub-markets}

We first consider the two sub-markets that consist of the three assets and each multi-asset derivative separately, \textit{i.e.} $(S^1, S^2, S^3, Z^1)$ and $(S^1, S^2, S^3, Z^2)$, and we are interested in deriving bounds for the arbitrage-free prices of $Z^1$ and $Z^2$. 
The functions $z_1$ and $-z_2$ are $\Delta$-monotonic, hence a lower and upper bound for $\Pi(Z^1)$ and $\Pi(Z^2)$ can be derived by the Fr\'echet--Hoeffding bounds; indeed, we have
\begin{align*}
\widehat{\pi}_{z_1}(\overline{W}_3)\le p\le\widehat{\pi}_{z_1}(\overline{M}_3),	&\ \text{ for every } p\in\Pi(Z^1)\,, \\
\widehat{\pi}_{z_2}(\overline{M}_3)\le p\le\widehat{\pi}_{z_2}(\overline{W}_3),	&\ \text{ for every } p\in\Pi(Z^2)\,.
\end{align*}
The support of the measures induced by $z_1$ and $z_2$ is one-dimensional and lies equally distributed along the diagonal, \textit{i.e.}
\begin{align*}
	\supp (\mu_{z_1}) &= \big\{ \mathbf{x}\in[K_1,\infty)^3 \,|\, x_1=x_2=x_3 \big\}\, , \\
	\supp (\mu_{z_2}) &= \big\{ \mathbf{x}\in[0,K_2]^3 \,|\, x_1=x_2=x_3 \big\} \,.
\end{align*}
Moreover, since $z_{1,I}\equiv 0$ and $z_{2,I}\equiv K_2$ for all $I$ with $|I|=1,2$, we get that $\mu_{z_{1,I}}=\mu_{z_{2,I}}=0$ . 
This also implies
\begin{align*}
	\int\limits_{\mathbb{R}_+} |z_{i,I}(x,x,x)| \ud F_{\mathbb Q}(\mathbf x) =0<\infty \,, \ \ i=1,2\,,
\end{align*}
while for $I=\{1,2,3\}$ we have
\begin{align*}
	\int\limits_{\mathbb{R}_+} |z_i(x,x,x)| \ud F_{\mathbb Q}(\mathbf x) = \mathbb{E}\big[ |z_i(S^1_1, S^2_1, S^3_1)| \big] < \infty \,, \ \ i=1,2\,.
\end{align*}

Hence, the expectation operator $\widehat{\pi}_{z_i}$ is well-defined for $i=1,2$ by \citet[Proposition 5.8]{lux2016}.
Let us also mention that $\mu_{z_1}$ and $-\mu_{z_2}$ are positive measures. 
Now, noting that $z_1(0,0,0)=0$ and $z_2(0,0,0)=K_2$, we deduce the followings bounds for $\Pi(Z^1)$ and $\Pi(Z^2)$:
\begin{align*}
	\widehat{\pi}_{z_1} \big(\overline{W}_3\big) &= \int\limits_{[K_1,\infty)} \overline{W}_3\big(F_1(x), F_2(x), F_3(x)\big) \ud x \,,\\
	\widehat{\pi}_{z_1} \big(\overline{M}_3\big) &= \int\limits_{[K_1,\infty)} \overline{M}_3\big(F_1(x), F_2(x), F_3(x)\big) \ud x\,,\\
	\widehat{\pi}_{z_2} \big(\overline{W}_3\big) &= K_2-\int\limits_{[0,K_2]} \overline{W}_3\big(F_1(x), F_2(x), F_3(x)\big) \ud x\,,\\
	\widehat{\pi}_{z_2} \big(\overline{M}_3\big) &= K_2-\int\limits_{[0,K_2]} \overline{M}_3\big(F_1(x), F_2(x), F_3(x)\big) \ud x\,.
\end{align*}
A numerical illustration of these bounds is depicted in Figures \ref{fig:PriceBoundsLognormal} and \ref{fig:PriceBoundsNIG} for the log-normal and the NIG distribution respectively.

\begin{figure}[h]
	\centering
	\includegraphics[scale=0.35]{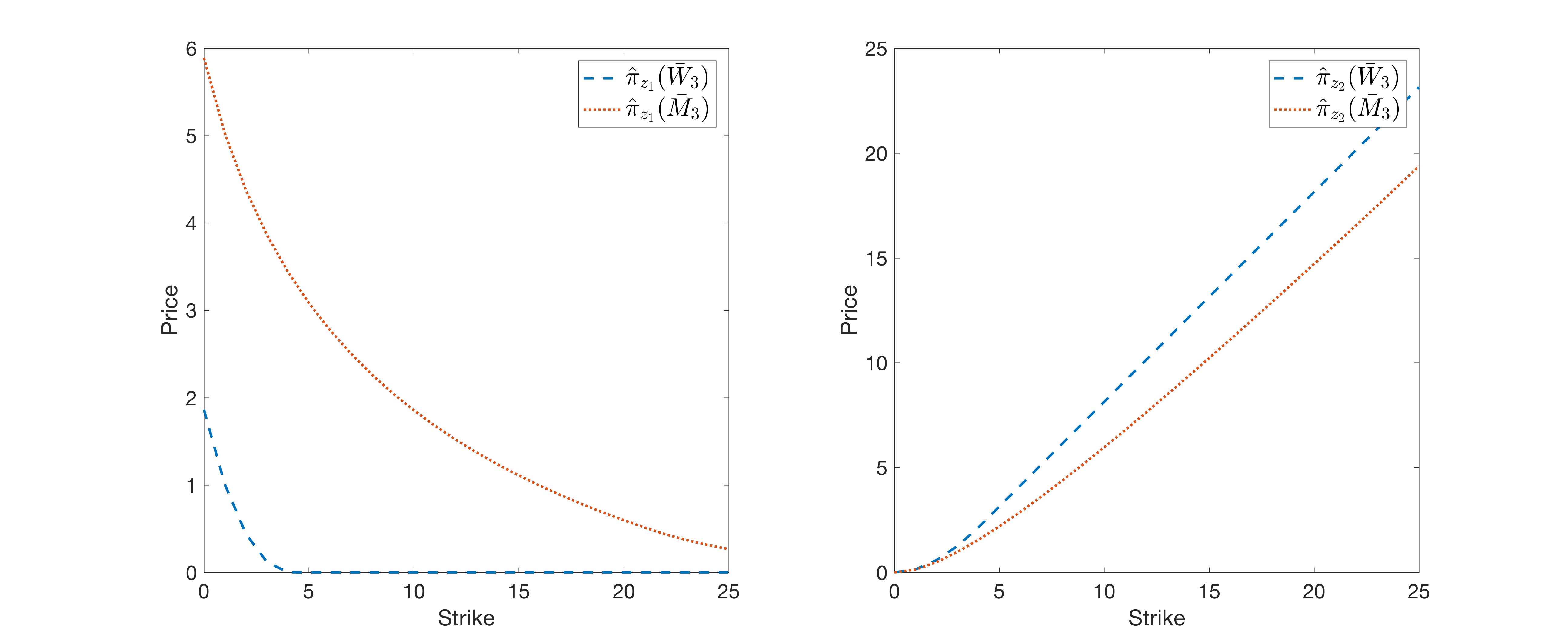}
	\caption{Bounds for $\Pi(Z^1)$ (left) and $\Pi(Z^2)$ (right) derived by the \FH bounds as a function of the strike; Setting 1.}
	\label{fig:PriceBoundsLognormal}
\end{figure}

\begin{figure}[h]
      \centering
      \includegraphics[scale=0.35]{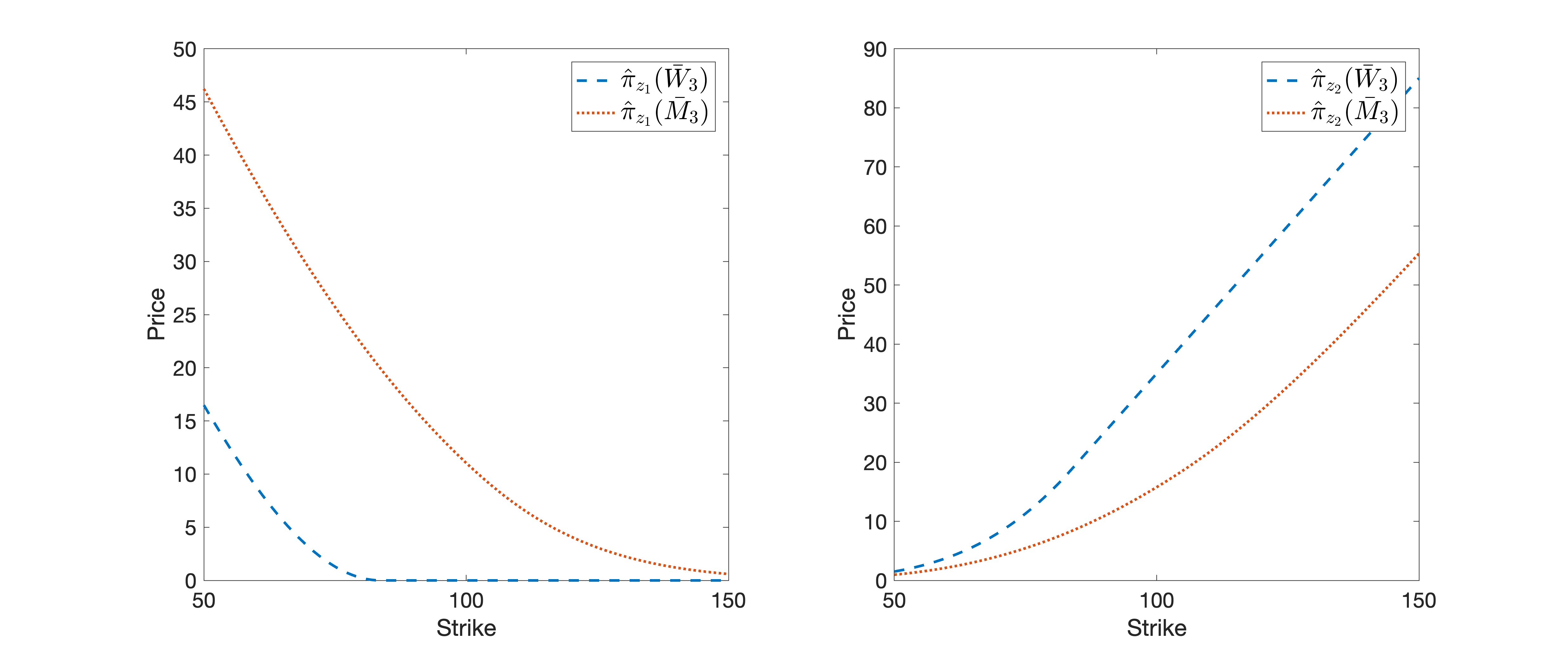}
      \caption{Bounds for $\Pi(Z^1)$ (left) and $\Pi(Z^2)$ (right) derived by the Fr\'echet--Hoeffding bounds as a function of the strike; Setting 2.}
      \label{fig:PriceBoundsNIG}
\end{figure}

\begin{remark}
Let us point out that the structure of the lower \FH bound, together with the properties of the measures induced by the payoff functions, lead to some restrictions for the price bounds of put and call options on the minimum of several assets.
Indeed, since $\overline{W}_d(\mathbf{u}) = W_d(\mathbf 1 - \mathbf u)$ is decreasing in every $u_i$ and floored by $0$, there exists an $x_0$ such that $\overline{W}_d\big(F_1(x), \dots, F_d(x)\big)=0$ for $x>x_0$. 
This $x_0$ depends on the marginal distributions. 
On the one hand, this yields for the lower bound of the call price that $\widehat{\pi}_{z_1}(\overline{W}_3)=0$ for $K_1\ge x_0$; see also Figure \ref{fig:PriceBoundsLognormal} (where the price bounds are plotted as functions of the strike). 
On the other hand, for the {upper} bound of the put price, we get that $\widehat{\pi}_{z_2}(\overline{W}_3) = K_2 + c$ for $K_2\ge x_0$ with constant $c\equiv\int_{[0,x_0]} \overline{W}_3\big(F_1(x), F_2(x), F_3(x)\big) \ud x$.
There is no equivalent statement for the other two bounds $\pi_{z_1}(\overline{M}_3)$ and $\pi_{z_2}(\overline{M}_3)$ because, in general, $F_i(x)<1$ for $x<\infty$.
\end{remark}


\subsection{Detecting an arbitrage}

Finally, we present an application of the main result of this work, \textit{i.e.} Theorem \ref{thm:arbitrage}. 
More specifically, we detect an arbitrage in the market $(S^1, S^2, S^3, Z^1, Z^2)$ that contains three assets and two three-asset derivatives, even though the prices of $Z^1$ and $Z^2$ lie inside their respective no-arbitrage bounds. 
\citet{Tav15} searches for the global minimum of the objective function $f_{\mathrm{\mathrm{obj}}}(\mathbf{u}):=\overline Q_{\mathbf{p}}(\mathbf{u})-\underline Q_{\mathbf{p}}(\mathbf{u})$ over the unit square. 
However, it suffices to find a $\mathbf{u}^*$ such that $f_{\mathrm{obj}}(\mathbf{u}^*)<0$ and not necessarily the global minimum. 
Since we consider an additional dimension, we restrict ourselves to checking whether $f_{\mathrm{obj}}$ becomes negative or not.

In Setting 1, we consider the call and put option on the minimum of three assets $Z^1$ and $Z^2$ with strikes $K_1=3$ and $K_2=8$ respectively. 
Then we have approximately the following no-arbitrage bounds:
	\begin{align*}
		\Pi(Z^1)=[0.118 , 3.864] \quad \text{and} \quad \Pi(Z^2)=[4.374 , 6.138] \,.
	\end{align*}
Assume that the traded price for the call equals $3.5$ and the traded price for the put equals $6$, \textit{i.e.} $\mathbf{p}=(3.5,6)$. 
Obviously both prices lie within their respective no-arbitrage bounds, hence the two sub-markets where either $Z^1$ or $Z^2$ is the only multi-asset derivative are free of arbitrage. 
However, we numerically compute that
	\begin{align*}
		f_{\mathrm{obj}}(0.7,0.5,0.1)\approx -0.0952<0 \,,
	\end{align*}
therefore Theorem \ref{thm:arbitrage} yields that the market with both multi-asset derivatives is \textit{not} free of arbitrage, \textit{i.e.} $\mathbf{p}\notin\Pi(Z^1,Z^2)$. 
Figure $\ref{fobjLognormal}$ shows a plot of the objective function $f_{\mathrm{obj}}$. 
One can see clearly how $f_{\mathrm{obj}}$ drops below zero around $\mathbf{u}=(0.7,0.5,0.1)$. 

In Setting 2, we consider again the call and put option on the minimum of three assets $Z^1$ and $Z^2$ with strikes $K_1=60$ and $K_2=140$ respectively. 
Then we have approximately the following no-arbitrage bounds:
       \begin{align*}
             \Pi(Z^1)=[8.799 , 37.431] \quad \text{and} \quad \Pi(Z^2)=[45.94 , 75.018] \,.
       \end{align*}
Assume that the traded price for the call equals $9$ and the traded price for the put equals $46.1$, \textit{i.e.} $\mathbf{p}=(9,46.1)$. 
Obviously both prices lie within their respective no-arbitrage bounds, hence the two sub-markets where either $Z^1$ or $Z^2$ is the only multi-asset derivative are free of arbitrage. 
However, we numerically compute that
       \begin{align*}
             f_{\mathrm{obj}}(0.4,0.3,0.4)\approx -0.4905<0 \,,
       \end{align*}
therefore Theorem \ref{thm:arbitrage} yields that the market with both multi-asset derivatives is \textit{not} free of arbitrage, \textit{i.e.} $\mathbf{p}\notin\Pi(Z^1,Z^2)$. 
Figure $\ref{fobjNIG}$ shows a plot of the objective function $f_{\mathrm{obj}}$. 
One can see clearly how $f_{\mathrm{obj}}$ drops below zero around $\mathbf{u}=(0.4,0.3,0.4)$.

An intuitive explanation behind the appearance of arbitrage for the price vectors $\mathbf{p}=(3.5,6)$, resp. $\mathbf{p}=(9,46.1)$, could be as follows:
The prices for $Z^1$ and $Z^2$ are both taken from the upper, resp. lower, part of the intervals $\Pi(Z^1)$ and $\Pi(Z^2)$; however, one payoff function is non-decreasing and the other is non-   increasing with respect to the upper orthant order, which diminishes the possibility of finding a copula $C$ such that both $\widehat{\pi}_{z_1}(\widehat{C})=p_1$ and $\widehat{\pi}_{z_2}(\widehat{C})=p_2$
. 
A similar result in Setting 1 appears if we choose both prices close to the lower bounds; for example, for $\mathbf{p}=(0.3, 4.5)$ we get that
\begin{align*}
	f_{\mathrm{obj}}(0.7,0.5,0.1)\approx -0.1257<0 \,.
\end{align*}
However, we have not noticed something analogous in Setting 2 when selecting both prices close to the upper bounds.
On the other hand, if we select a price away from the upper bound for $Z^2$, \textit{e.g.} $\mathbf{p}=(3.5, 4.5)$ in Setting 1, then the objective function does not become negative any longer.
Indeed, we find that the global minimum of the objective function $f_{\mathrm{obj}}$ is zero, and is attained for $u_i=0$ or $u_i=1$ for some $i=1,2,3$, \textit{i.e.} on the boundaries of the unit cube $[0,1]^3$. 
Let us point out again that this does not necessarily imply that the market is free of arbitrage, since Theorem \ref{thm:arbitrage} only provides a necessary condition.

\begin{figure}[h!]
	\begin{subfigure}[t]{0.03\textwidth}
		\textrm{(a)}
	\end{subfigure}
	\begin{subfigure}[t]{0.3\textwidth}
		\centering
		\includegraphics[scale=0.33]{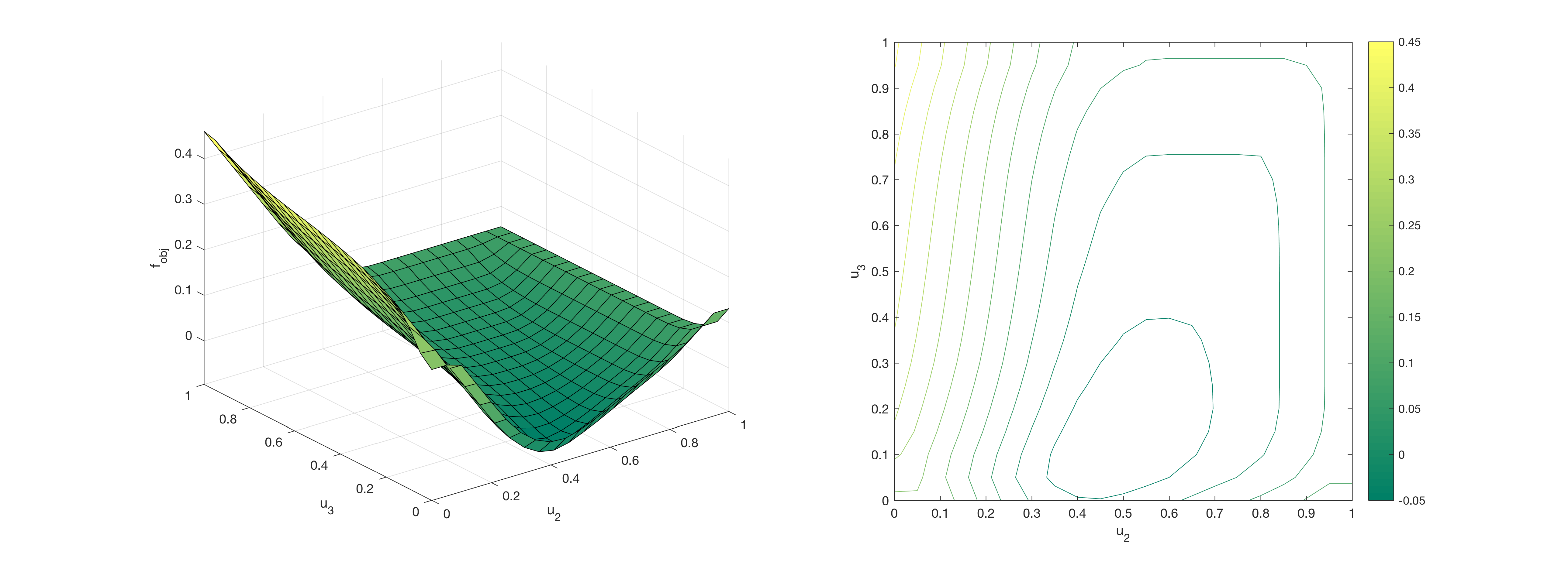}
	\end{subfigure}

	\begin{subfigure}[t]{0.03\textwidth}
		\textrm{(b)}
	\end{subfigure}
	\begin{subfigure}[t]{0.3\textwidth}
		\centering
		\includegraphics[scale=0.33]{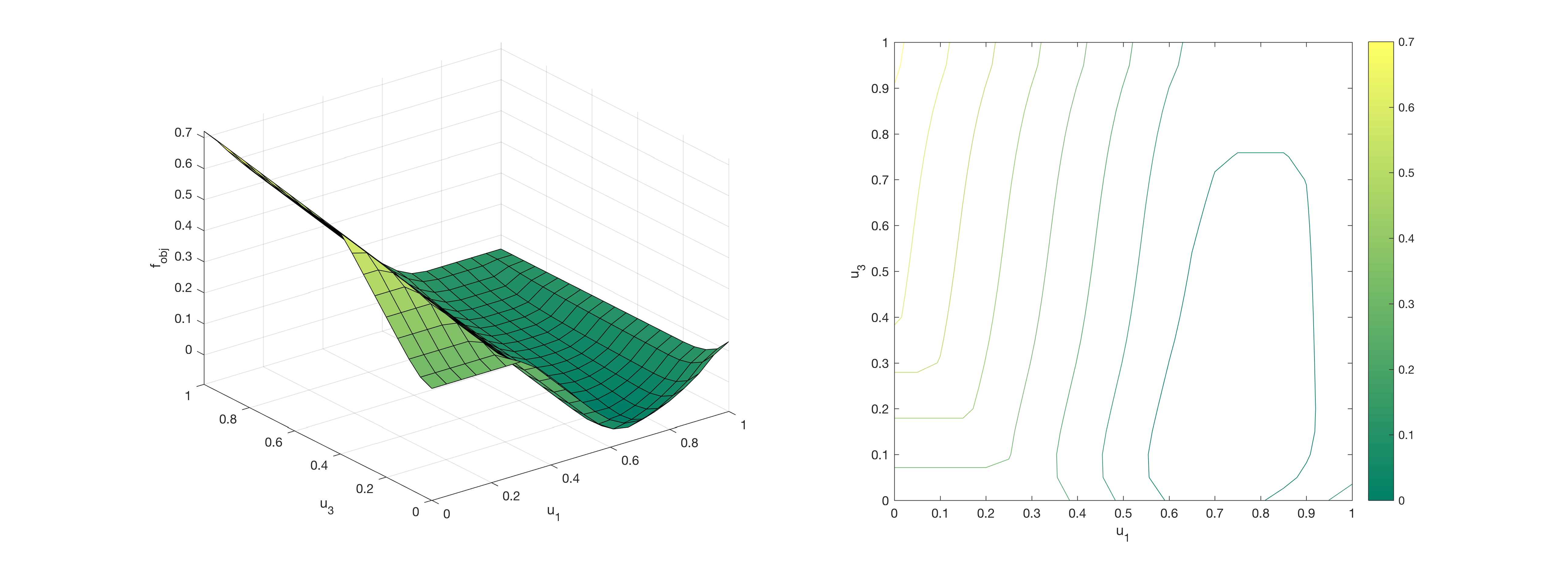}
	\end{subfigure}
	
	\begin{subfigure}[t]{0.03\textwidth}
		\textrm{(c)}
	\end{subfigure}
	\begin{subfigure}[t]{0.3\textwidth}
		\centering
		\includegraphics[scale=0.33]{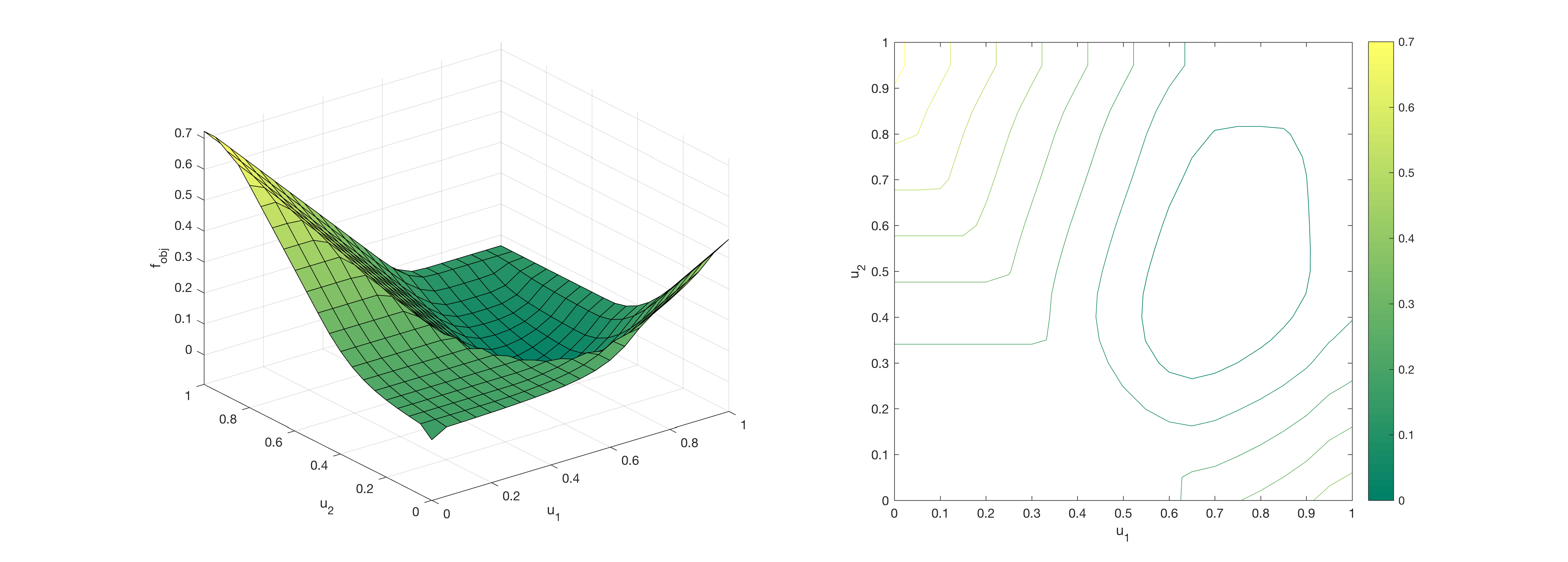}
	\end{subfigure}
	\caption{Values and contour plots of the objective function $f_{\mathrm{obj}}$ for $\mathbf{p}=(3.5,6)$ and for the marginals restricted on (a) $u_1=0.7$, (b) $u_2=0.5$, (c) $u_3=0.1$; Setting 1.}
	\label{fobjLognormal}
\end{figure}

\begin{figure}[h!]
      \begin{subfigure}[t]{0.03\textwidth}
            \textrm{(a)}
      \end{subfigure}
      \begin{subfigure}[t]{0.3\textwidth}
            \centering
            \includegraphics[scale=0.33]{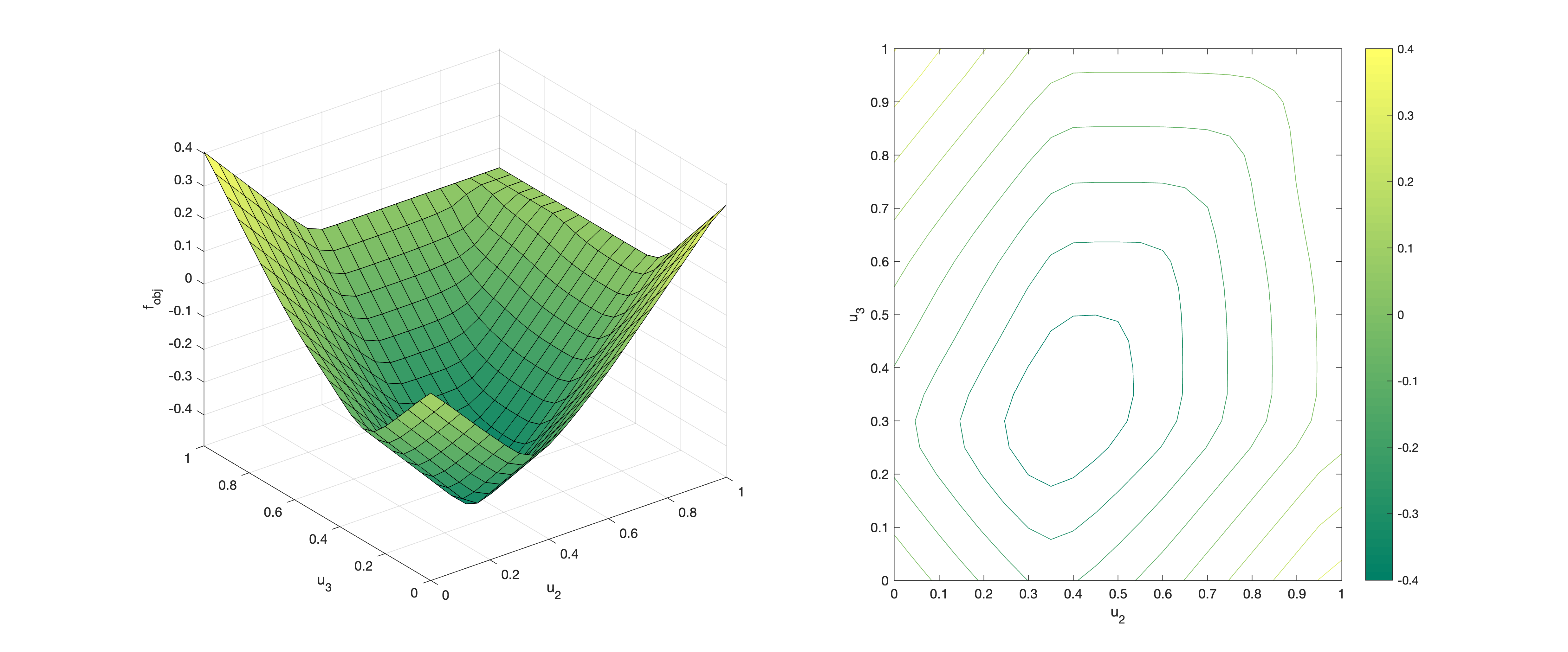}
      \end{subfigure}

      \begin{subfigure}[t]{0.03\textwidth}
            \textrm{(b)}
      \end{subfigure}
      \begin{subfigure}[t]{0.3\textwidth}
            \centering
            \includegraphics[scale=0.33]{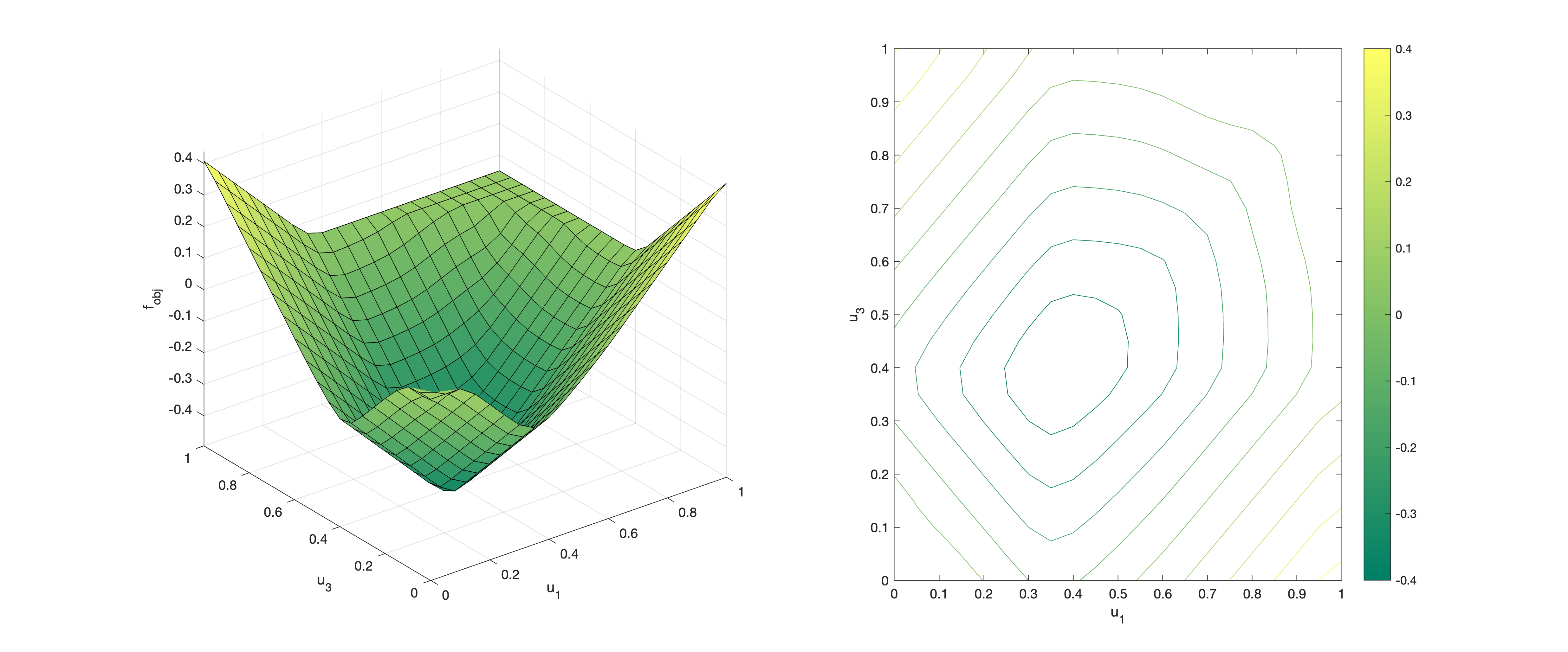}
      \end{subfigure}
      
      \begin{subfigure}[t]{0.03\textwidth}
            \textrm{(c)}
      \end{subfigure}
      \begin{subfigure}[t]{0.3\textwidth}
            \centering
            \includegraphics[scale=0.33]{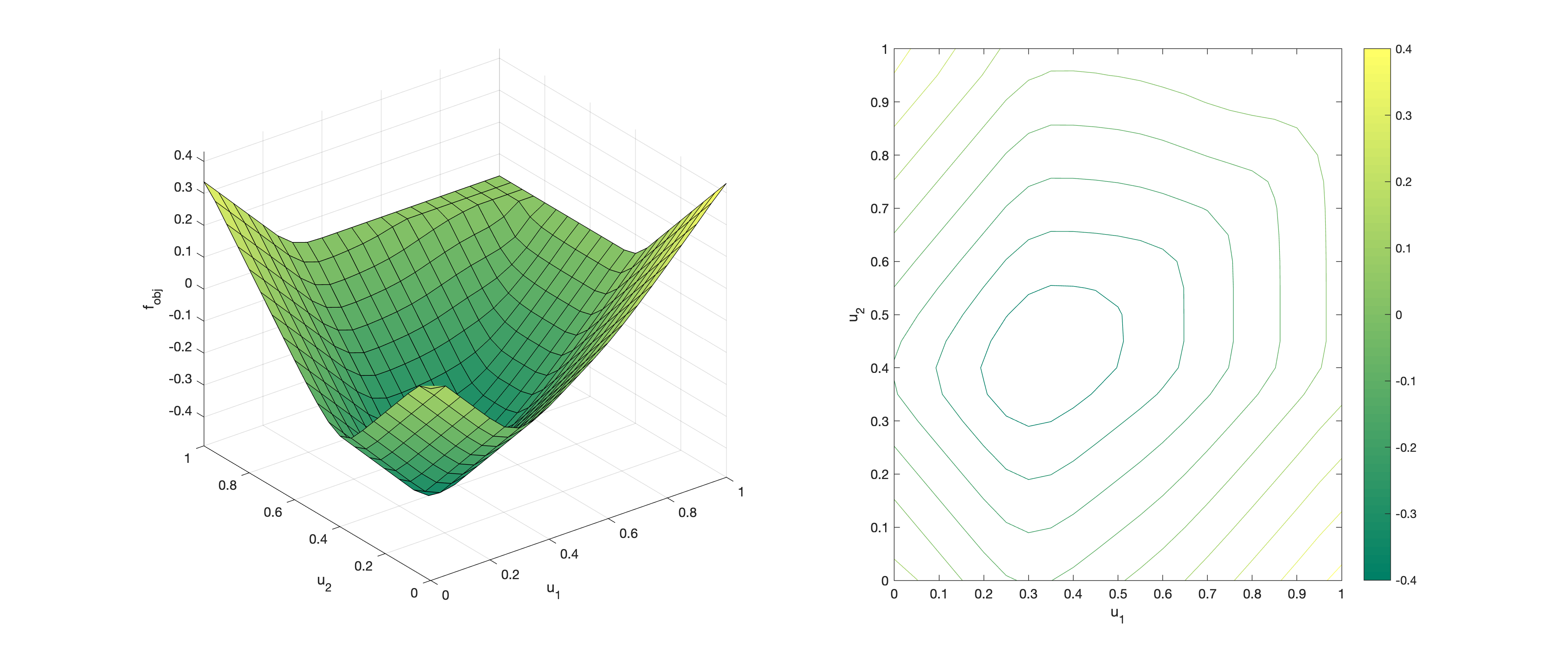}
      \end{subfigure}
      \caption{Values and contour plots of the objective function $f_{\mathrm{obj}}$ for $\mathbf{p}=(9,46.1)$ and for the marginals restricted on (a) $u_1=0.4$, (b) $u_2=0.3$, (c) $u_3=0.4$; Setting 2.}
      \label{fobjNIG}
\end{figure}


\appendix

\section{Improved \FH bounds for non-increasing functionals}
\label{app:B}

The following two theorems cover the case when the map $\rho$ is non-increasing with respect to the orthant orders. 
This appears in our work when the negation of the payoff function, say $-\rho$, is either $\Delta$-monotonic or $\Delta$-antitonic. 
In that case, we get that $\rho(M_d)\le\rho(W_d)$. 
The proofs of these results are omitted for the sake of brevity, as they are completely analogous to the proofs of Theorems 3.3 and A.2 in \citet{lux2016}.

\begin{theorem}
Let $\rho:\mathcal{Q}^d\rightarrow\mathbb{R}$ be non-increasing with respect to the lower orthant order and continuous with respect to the pointwise convergence of quasi-copulas. 
Let $\theta\in[\rho(M_d),\rho(W_d)]$  and define
\begin{align*}
	\mathcal{Q}^{\rho,\theta} := \big\{ Q\in\mathcal{Q}^d \,|\, \rho(Q)=\theta \big\} \,.
\end{align*}
Then, for all $Q\in\mathcal{Q}^{\rho,\theta}$, holds
\begin{align*}
	Q_L^{\rho,\theta}(\mathbf{u})\le Q(\mathbf{u})\le Q_U^{\rho,\theta}(\mathbf{u}) \quad\textit{for all }\mathbf{u}\in[0,1]^d\,,
\end{align*}
with
\begin{align*}
Q_L^{\rho,\theta}(\mathbf{u})
	:=&\begin{cases} 
		\rho^{-1}_+(\mathbf{u},\theta)\,,\quad \textit{if}\,\, \theta\in[\rho(M_d),\rho_+(\mathbf{u}, W_d(\mathbf{u}))], \\
		M_d(\mathbf{u})\,, \quad\quad \textit{otherwise}\,,
	\end{cases} \\
Q_U^{\rho,\theta}(\mathbf{u})
	:=&\begin{cases} 
		\rho^{-1}_-(\mathbf{u},\theta)\,,\quad \textit{if}\,\, \theta\in[\rho_-(\mathbf{u},M_d(\mathbf{u})),\rho(W_d)], \\
		W_d(\mathbf{u})\,, \quad\quad \textit{otherwise}.
	\end{cases}
\end{align*}
\end{theorem}

\begin{theorem}
Let $\rho:\mathcal{C}^d\rightarrow\mathbb{R}$ be non-increasing with respect to the upper orthant order and continuous with respect to the pointwise convergence of copulas. 
Let $\theta\in[\widehat{\rho}(\overline{M}_d),\widehat{\rho}(\overline{W}_d)]$ and define
\begin{align*}
	\widehat{\mathcal{C}}^{\rho,\theta} := \big\{ C\in\mathcal{C}^d \,|\, \widehat{\rho}(\widehat{C})=\theta \big\}.
\end{align*}
Then, for all $\widehat{C}\in\widehat{\mathcal{C}}^{\rho,\theta}$, holds
\begin{align*}
	\widehat{Q}_L^{\rho,\theta}(\mathbf{u})\le\widehat{C}(\mathbf{u})\le\widehat{Q}_U^{\rho,\theta}(\mathbf{u}) \quad\textit{for all}\,\,\mathbf{u}\in[0,1]^d\,,
\end{align*}
with
\begin{align*}
\widehat{Q}_L^{\rho,\theta}(\mathbf{u})
	:=&\begin{cases} 
		\widehat{\rho}^{\,-1}_+(\mathbf{u},\theta)\,,\quad \textit{if}\,\, \theta\in[\widehat{\rho}(\overline{M}_d),\widehat{\rho}_+(\mathbf{u}, \overline{W}_d(\mathbf{u}))], \\
		\overline{M}_d(\mathbf{u})\,, \quad\quad \textit{otherwise}\,,
	\end{cases} \\
\widehat{Q}_U^{\rho,\theta}(\mathbf{u})
	:=&\begin{cases} 
		\widehat{\rho}^{\,-1}_-(\mathbf{u},\theta)\,,\quad \textit{if}\,\, \theta\in[\widehat{\rho}_-(\mathbf{u},\overline{M}_d(\mathbf{u})),\widehat{\rho}(\overline{W}_d)], \\
		\overline{W}_d(\mathbf{u})\,, \quad\quad\textit{otherwise.}
	\end{cases}
\end{align*}
\end{theorem}








\end{document}